\documentclass{article}

\usepackage[hmarginratio=3:2]{geometry}
\usepackage[table]{xcolor} 
\usepackage{tikz}
\usetikzlibrary{matrix,fit,decorations.pathreplacing}

\usepackage[utf8]{inputenc} 
\usepackage{xcolor}
\usepackage{lmodern}
\usepackage{amsmath}
\usepackage[colorlinks=true,linkcolor=blue,citecolor=blue]{hyperref}
\usepackage[linesnumbered,vlined,commentsnumbered, ruled,noend]{algorithm2e}
\usepackage{amssymb}
\usepackage{enumerate}
\usepackage{enumitem}
\usepackage{nicefrac}
\usepackage{multirow}
\usepackage{dsfont}
\usepackage{tabularx}
\usepackage{aligned-overset}
\usepackage{xurl}

\usepackage[]{hyperref}

\usepackage{booktabs}
\usepackage{multirow}

\definecolor{light-gray}{gray}{0.95}

\usepackage{amsthm}
\theoremstyle{definition}
\newtheorem{definition}{Definition}[section]
\theoremstyle{plain}

\newtheorem{proposition}{Proposition}[section]

\newtheorem{lemma}{Lemma}[section]

\newtheorem{rrule}{Rule}[section]

\newcommand{\Oh}{\ensuremath{\mathcal{O}}}

\newcommand{\problemdef}[3]{
  \begin{center}
    \begin{minipage}{0.95\textwidth}
      \normalsize\textsc{#1} \smallskip \\
      \begin{tabularx}{\textwidth}{@{}l@{\hspace{3pt}}X}
        \normalsize\textbf{Input:}    & \normalsize#2 \\
        \normalsize\textbf{Task:} & \normalsize#3
      \end{tabularx}
    \end{minipage}
  \end{center}
}

\newcommand{\dist}{\text{dist}}

\begin{document}

\title{Efficient Branch-and-Bound Algorithms for Finding Triangle-Constrained 2-Clubs\footnote{Preliminary results and an implementation for the special case of~$\ell=1$ for \textsc{Vertex Triangle~$2$-Club} appeared in the second author's Bachelor thesis~\cite{Ke21}. }}

\author{Niels Grüttemeier 
\and Philipp~Heinrich~Keßler
 \and Christian Komusiewicz 
  \and Frank Sommer\footnote{Supported by the Deutsche Forschungsgemeinschaft (DFG), project EAGR, {KO~3669/{6-1}}.} 
 }
\date{%
     Fachbereich Mathematik und Informatik, Philipps-Universität Marburg, Germany\\[2ex]%
}






\maketitle

\begin{abstract}
In the \textsc{Vertex Triangle 2-Club} problem, we are given an undirected graph~$G$ and aim to find a maximum-vertex subgraph of~$G$ that has diameter at most~2 and in which every vertex is contained in at least~$\ell$~triangles in the subgraph. 
  So far, the only algorithm for solving \textsc{Vertex Triangle 2-Club} relies on an ILP formulation [Almeida and Br{\'{a}}s, Comput. Oper. Res. 2019]. 
  In this work, we develop a combinatorial branch-and-bound algorithm that, coupled with a set of data reduction rules, outperforms the existing implementation and is able to find optimal solutions on sparse real-world graphs with more than 100\,000 vertices in a few minutes. 
  We also extend our algorithm to the \textsc{Edge Triangle 2-Club} problem where the triangle constraint is imposed on all edges of the subgraph.
\end{abstract}

\section{Introduction}
One of the characteristic features of communities in social networks is that they have very low diameter: each pair of vertices in a community is connected via a very small number of hops. This observation motivated the proposal of diameter-based community models; the arguably most famous one being the $s$-club model of Mokken~\cite{Mok79}. In this model, a vertex set~$S$ in a graph~$G$ is considered to be a community if~$G[S]$, the subgraph of~$G$ induced by~$S$, has low diameter. Informally, this means that every pair of vertices in~$S$ is connected via few hops that visit only vertices of~$S$. The corresponding optimization problem of finding a largest community under this model is the NP-hard~\textsc{$s$-Club} problem~\cite{BLP02,BBT05}. For~$s=1$, this problem is the same as \textsc{Maximum Clique}.

A particularly important case is~$s=2$, where we search for a set that induces a subgraph of diameter at most 2. While it has been shown that~\textsc{$s$-Club} can be solved efficiently by ILP formulations~\cite{BBT05,PBH16,SB20} and by combinatorial branch-and-bound algorithms~\cite{BBT05,CHLS13,HKN15,KNNP19}, these studies showed an undesirable behavior of maximum 2-clubs: in most instances, the largest 2-club consists of the the vertex of maximum degree and its neighbors. As a consequence, maximum 2-clubs fare poorly when it comes to other typical properties of communities such as having high density or being robust against vertex or edge failures. 

To overcome these drawbacks of the 2-club (or more generally, the $s$-club) model, many augmented variants of 2-clubs have been proposed~\cite{AB19,CA17,KNNP19,PYB13,VB12,YPB17}. One such variant, the triangle 2-club model, asks that every vertex in the 2-club is part of a triangle~\cite{CA17}. This model was later extended to one that may ask for any fixed number of triangles. More precisely, a vertex set~$S\subseteq V$ is a \emph{vertex-$\ell$-triangle 2-club} in~$G$ if~$G[S]$ has diameter at most two and every vertex of~$S$ is in at least~$\ell$ triangles in~$G[S]$~\cite{AB19}. 

\problemdef{Vertex Triangle~$2$-Club} {An undirected graph~$G$ = $(V,E)$ and
  an integer~$\ell$.} {Find a maximum-cardinality vertex-$\ell$-triangle 2-club.}

One desirable effect of the~$\ell$-triangle 2-club property that is not guaranteed by other augmented 2-club models concerns the local clustering coefficients of the graph~$G[S]$. The local clustering coefficient of a vertex~$v$ in a graph~$G$ is the ratio of present edges in the neighborhood~$N(v)$ of~$v$ and~$\binom{\deg(v)}{2}$, the number of possible edges in~$N(v)$. Roughly speaking, the larger one sets~$\ell$, the better the minimum clustering coefficient in the solution is~\cite{AB19}. Other augmented 2-club models do not give this guarantee since their definitions are usually fulfilled by complete bipartite graphs~\cite{KNNP19,PYB13,VB12}.

A further related model,  introduced in companion work~\cite{GKS22}, imposes the triangle constraints on the \emph{edges} of the~$s$-club. 
More precisely, we say that a set~$S$ is an \emph{edge-$\ell$-triangle 2-club} when~$G$ has a spanning subgraph~$\widehat{G}=(S,\widehat{E})$ such that~$\widehat{G}$ has diameter at most~2 and every edge of~$\widehat{E}$ is in at least~$\ell$ triangles in~$\widehat{G}$.
\problemdef{Edge Triangle~$2$-Club}
{An undirected graph~$G$ = $(V,E)$ and an integer~$\ell$.}
{Find a maximum-cardinality edge-$\ell$-triangle 2-club.} 

Observe that every set that fulfills the edge-$\ell$-triangle property also fulfills the vertex-$\ell$-triangle property while there exist vertex-$\ell$-triangle 2-clubs that do not satisfy the edge-$\ell$-triangle property~\cite{GKS22}. From a theoretical point of view, \textsc{Vertex Triangle~$2$-Club} and \textsc{Edge Triangle~$2$-Club} are NP-hard~\cite{AB19,CA17,GKS22} and, moreover, both problems are considerably harder than~\textsc{2-Club} with respect to their parameterized complexity for the parameter solution size~$|S|$~\cite{GKS22}. 

In this work, we study whether both problems can be solved efficiently in practice. For \textsc{Vertex Triangle 2-Club}, a first ILP-based implementation showed that it can be solved quickly on medium-size instances for all~$\ell \le 6$~\cite{AB19}. The ILP-based algorithms, which more generally solve \textsc{Vertex Triangle $s$-Club} for arbitrary values of~$s$, consist of three main features: 1) an efficient separation that computes valid inequalities in case a vertex set violates the triangle-property or the diameter property, 2) a data reduction rule that deletes vertices which are not in sufficiently many triangles, and 3) an efficiently computable lower bound, called the \emph{Neighborhood Lower Bound} in this work. The \emph{Neighborhood Lower Bound} is an adaption of the classic star heuristic for 2-club~\cite{AB19}.  The star heuristic finds the vertex with maximum degree as a valid 2-club solution which is in most cases the optimal solution. 

\paragraph{Our Results.}
For brevity, we use the term triangle 2-clubs to simultaneously refer to  $\ell$-vertex and the $\ell$-edge triangle properties for all~$\ell$.

We show that, using a combinatorial branch-and-bound algorithm combined with problem-specific data reduction rules and lower bounds, allows us to solve large sparse instance of \textsc{Vertex Triangle~$2$-Club} with up to~$290\,000$ vertices and~$990\,000$~edges. 
Compared to the previous computational experiments, which were run only for~$\ell\le 6$, we show that the problem can be solved efficiently for~$\ell$ up to 100. 
Furthermore, we show that our implementation outperforms the previous ILP. Our implementation also solves the \textsc{Edge Triangle~$2$-Club}. 

The overall approach of the branch-and-bound algorithm is similar to previous ones for \textsc{2-Club}~\cite{HKN15} and robust models of \textsc{2-Club}~\cite{KNNP19}: First, perform data reduction and compute a lower bound, then consider the sufficiently large 2-neighborhoods of~$G$ one by one. 
Furthermore, on each 2-neighborhood, perform a branching that identifies a vertex that is in conflict with at least one other vertex, for example since their distance is too large. We also provide a running time bound for this algorithm in terms of~$n$, the number of vertices of the graph and~$n-|S|$, the number of vertices that are not part of triangle 2-club.

Naturally, the details of the branch-and-bound algorithm differ from previous ones when it comes to incorporating the triangle properties. To establish the triangle properties, we use rules that delete vertices or edges whenever they are not in sufficiently many triangles. These rules have been described previously~\cite{AB19,GKS22}. To speed up the data reduction, we add further simple-degree-based rules. In addition, we identify new data reduction rules that delete edges and vertices whose inclusion in a triangle 2-club would lead to a violation of the triangle properties. These rules particularly apply to the setting during branching when some vertices are already marked as being part of the sought triangle 2-club. 

We use two heuristics for computing lower bounds. The first one is the known \emph{Neighborhood Lower Bound}. This lower bound performs very well on many instances but not on all instances. In particular for intermediate values of~$\ell$, the lower bound quality decreases. This prompts us to develop a new heuristic, called \emph{Greedy Lower Bound} which considers the 2-neighborhoods, as the branching algorithm does. Instead of performing a branching it greedily deletes vertices that are in a conflict.  We show that the combination of the two lower bounds substantially outperforms the \emph{Neighborhood Lower Bound} in terms of solution quality. Since the running time of the heuristic is quite high it leads to substantial running time improvements only for \textsc{Edge Triangle~$2$-Club}. Besides the experimental evaluation we also analyze the theoretical worst-case running time of the lower bounds depending on the degeneracy, the maximum degree, and the total size of the input graph.

We then study how the parameter~$\ell$ influences central properties of the returned triangle 2-clubs, as  done in previous work~\cite{AB19}. In a nutshell, we confirm that increasing~$\ell$ usually leads to high density and high global and local clustering coefficients, also on larger networks, where the general tendency is that (triangle) 2-clubs are less dense. Since increasing~$\ell$ usually does not influence the running time negatively, this gives a good adjustable parameter for balancing the emphasis on the low-diameter property with other cohesiveness properties.

\paragraph{Preliminaries.}
\label{sec:prelim}
We consider undirected simple graphs~$G$ and let~$V(G)$ and $E(G)$ denote the vertex set
and edge set of~$G$, respectively. We denote~$n:=|V(G)|$ and~$m:=|E(G)|$. The
\emph{subgraph of~$G$ induced by a vertex set~$S$} is denoted by~$G[S]$.  For a vertex~$v$, we let~$N(v)$ denote the (open) neighborhood of~$v$ and~$N[v]:=N(v)\cup \{v\}$ the closed neighborhood of~$v$. Moreover, we let~$N_2[v]$ denote the vertex set that contains~$N[v]$ and all vertices that have at least one common neighbor with~$v$. The \emph{degree}
of a vertex in~$G$ is denoted by~$\deg(v):=|N(v)|$. The \emph{diameter} of a graph~$G$ is the smallest
number~$t$ such that every pair of vertices in~$G$ is connected by a path of length at
most~$t$.  We let~$\Delta(G)$ denote the maximum degree of a graph~$G$. Moreover,~$d(G)$
denotes the \emph{degeneracy} of~$G$, that is, the smallest number~$d$ such that every
induced subgraph of~$G$ contains a vertex of degree at most~$d$. If the graph~$G$ is clear
from context, we simply write~$\Delta$ and~$d$. Any~$d$-degenerate~$G=(V,E)$ admits an
ordering~$\sigma=(v_1,\ldots,v_n)$ of~$V$ such that~$v_i$ has at most~$d$ neighbors
in~$\{v_{i+1},\ldots,v_n\}$. The ordering and the list of neighbors of all~$v_i$
in~$\{v_{i+1},\ldots,v_n\}$~can be computed in linear time.


\section{The Branching Algorithm}
\label{sec-bran}

Our algorithms which solve \textsc{Vertex Triangle 2-Club} and \textsc{Edge Triangle 2-Club} are based on the observation that the solution is contained in the 2-neighborhood
$N_2[v]$ of some vertex~$v$ of the input graph~\cite{SKMN12}. 
Consequently, we solve a given instance~$(G,\ell)$ by finding optimal solutions of all local instances~$(G_v,\ell)$, where~$G_v$ is the subgraph of~$G$ induced by~$N_2[v]$, and returning the solution of maximum size.

To solve the local instances~$(G_v,\ell)$, we use a branching algorithm that makes use of several reduction rules.

\subsection{Basic Reduction Rules and Marked Branching}

We first describe some reduction rules for establishing the triangle property.
Let~$(G,\ell)$ be an instance of \textsc{Vertex Triangle 2-Club} and let~$u$ be a vertex of~$G$. 
If~$u$ is contained in less than~$\ell$ triangles of~$G$, then~$u$ cannot be part of any vertex-$\ell$-triangle~2-club~$S$. 
Consequently,~$u$ can be safely deleted from the input graph. This deletion is also safe for \textsc{Edge Triangle 2-Club}, since a set that fulfills the edge-$\ell$-triangle property also fulfills the vertex-$\ell$-triangle property.
Moreover, for an instance~$(G,\ell)$  of~\textsc{Edge Triangle 2-Club} we can additionally delete some edges: Let~$e$ be an edge of~$G$ that is contained in less than~$\ell$ triangles. 
Furthermore, let~$S$ be an edge-$\ell$-triangle~$2$-club in~$G$ and let~$\widehat{G}=(S,\widehat{E})$ be a corresponding subgraph with diameter at most~$2$ where every edge of~$\widehat{E}$ is in at least~$\ell$ triangles in~$\widehat{G}$. 
Then,~$e \not \in \widehat{E}$ and therefore, the edge~$e$ can safely be deleted from the input instance.

Both observations described above lead to the following reduction rule.

\begin{rrule}[Low-Triangle Rule (LTR)]
\label{rr-remove-vertices-to-low}
\begin{enumerate}[label=$\alph*)$]
\item In case of \textsc{Vertex Triangle 2-Club}, delete all vertices that are in less than~$\ell$ triangles.
\item In case of \textsc{Edge Triangle 2-Club}, delete all edges that are in less than~$\ell$ triangles and delete isolated vertices.
\end{enumerate}
\end{rrule}
We now bound the running time of exhaustively applying this rule. To obtain a bound that explains why the rule runs relatively fast on sparse instances, we formulate it in terms of the number of edges~$m$ and the degeneracy~$d$ of the input graph. 
\begin{lemma}\label{lem:ltr-time}
  The LTR can be performed exhaustively in~$\Oh(m\cdot d)$ time. 
\end{lemma}
\begin{proof}
Using the triangle enumeration algorithm of Chiba and Nishizeki~\cite{CN85}, we can enumerate all triangles in $\Oh(m\cdot d)$~time.  While enumerating the triangles, we can compute for each vertex~$w$ and each edge~$e$ of~$G$ a list of all triangles containing~$w$ or~$e$ and their number.

After this prepossessing, a vertex or an edge to which the rule applies can be determined in $\Oh(1)$~time as long as the triangle counters for the vertices or edges are updated after each deletion. To update the triangle counter, we traverse the triangle list for the deleted vertex or edge and decrement the respective counter for all vertices or edges that are contained in the triangle.  
\end{proof}
To decrease the effort of counting triangles, we provide another rule that exploits the degree of vertices in the input graph to identify some vertices that are not contained in at least~$\ell$ triangles. Observe that a vertex~$v$ is contained in at most~$\binom{\deg(v)}{2}$ triangles, which is the case if~$N[v]$ is a clique. 
Therefore, a vertex~$v$ with~$\binom{\deg(v)}{2}< \ell$ is not contained in at least~$\ell$ triangles and we may thus delete vertices with degree at most~$ 1/2+\sqrt{1/4+2\ell}$ from the input graph. Furthermore, an edge~$e$ incident with a vertex~$v$ is contained in at most~$\deg(v)-1$ triangles. Thus, in case of \textsc{Edge Triangle 2-Club} we may delete vertices with degree at most~$\ell$.

\begin{rrule}[Low-Degree Rule (LDR)]
\label{rr-remove-vertices-to-low-degree}
\begin{enumerate}[label=$\alph*)$]
\item In case of \textsc{Vertex Triangle 2-Club}, delete all vertices that have degree at most~$1/2+\sqrt{1/4+2\ell}$.
\item In case of \textsc{Edge Triangle 2-Club}, delete all vertices that have degree at most~$\ell$.
\end{enumerate}
\end{rrule}
The LDR can be applied exhaustively in~$\Oh(m)$~time and thus it is faster than the LTR. 
Moreover, applying the LDR (Rule~\ref{rr-remove-vertices-to-low-degree}) may reduce the number of vertices and edges that need to be considered when applying the LTR (Rule~\ref{rr-remove-vertices-to-low}).
Note that these two rules were already observed by Almeida and Brás~\cite{AB19} for \textsc{Vertex Triangle 2-Club}.
After applying these two rules exhaustively, the input graph satisfies the triangle conditions posed on~vertex-$\ell$-triangle 2-clubs or edge-$\ell$-triangle~$2$-clubs. 
However, the diameter might still be larger than~$2$. To handle this, we use a branching strategy which is based on the definition of incompatibility. 

\begin{definition}
\label{def-compatible}
Two vertices~$u$ and~$w$ in a graph~$G$ are called \emph{compatible} if and only if~$\dist_G(u,w)\le 2$.
Otherwise,~$u$ and~$w$ are called \emph{incompatible}.
\end{definition}




Clearly, if all vertices are compatible, then the 2-club property is established and if there are two incompatible vertices, then one of the two must be deleted. This gives rise to a branching algorithm where we branch into two cases, deleting one vertex in each case. We formulate the branching in a slightly different manner by using a set~$M$ of \emph{marked vertices}, which was also done for other \textsc{2-Club} variants~\cite{HKN15,KNNP19}. Intuitively, the set~$M$ of marked vertices must belong to the solution that we aim to compute. 
Thus, in the presence of incompatible vertices~$u$ and~$w$, we may branch into one case where we assume that~$u$ is not part of the solution and one case where we assume that~$u$ is part of the solution. In the first case, we may delete~$u$; in the second case, we add~$u$ to~$M$. We may additionally delete~$w$ in the second case since~$u$ is now marked there is no 2-club that simultaneously contains~$u$ and~$w$. This observation is of course not only true for~$u$ and~$w$ but for any pair of incompatible vertices, where at least one of them is marked. This gives rise to the following two reduction rules.  

\begin{rrule}[Incompatible-Resolution Rule (IRR)]
\label{rr-remove-incompatible}
Delete all unmarked vertices that are incompatible to some marked vertex.
\end{rrule}

\begin{rrule}[Marked-Incompatible Rule (MIR)]
\label{rr-trivial-no}
If an instance contains two incompatible vertices that are both marked,
then return that there is no solution.
\end{rrule}
With these altogether four reduction rules and the marked branching, we can formulate our branching algorithm: apply the reduction rules exhaustively and perform the marked branching if there are at least two incompatible vertices. The pseudocode is given in Algorithm~\ref{algo-marked-branching}. The solution to the initial instance~$(G,\ell)$ where no vertex is marked can be found by calling the algorithm with~$M=\emptyset$.
Note that this branching algorithm also implies an FPT-algorithm for the parameter~$n-k$ for the decision version of the two problems, where we ask whether~$G$ has a triangle 2-club of size at least~$k$.

\begin{proposition}
\label{prop-vertex-and-edge-triangle-fpt-dual}
In $\Oh(2^{n-k}nm)$~time, we can decide whether~$G$ contains a~vertex-$\ell$-triangle 2-club~$S$ (an edge-$\ell$-triangle 2-club~$S$) of size at least~$k$, and return~$S$ if this is the case. 
\end{proposition}
\begin{proof}
  By Lemma~\ref{lem:ltr-time}, we may exhaustively apply the LDR and the LTR in $\Oh(md)=\Oh(nm)$~time. 
  We can also compute all pairs of incompatible vertices of~$G$ in~$\Oh(nm)$ time. 
  Thus, each application of the IRR, the MIR and the test whether further branching is necessary can be performed in~$\Oh(nm)$ time. This bounds the running time for each search tree node. After each application of the rules and after each recursive call created by the branching, the number of vertices is decreased by at least one, either directly (in the first case) or since the IRR is triggered (in the second case). Thus, at nodes of depth~$n-k$ we may conclude that the branch does not give a solultion of size~$k$. Hence, the depth of the search tree is at most~$n-k$ and the total search tree size is~$\Oh(2^{n-k})$. 
\end{proof}

\begin{algorithm}[t]
\KwIn{An instance~$(G=(V,E),\ell)$ and a set of marked vertices~$M \subseteq V$.}
\KwOut{A maximum-cardinality triangle 2-club~$S$ with~$M \subseteq S$.}

Apply Rules~\ref{rr-remove-vertices-to-low-degree}, \ref{rr-remove-vertices-to-low}, \ref{rr-remove-incompatible},~\ref{rr-trivial-no},~\ref{rr-remove-vertices-too-small-2-neighborhood}, and~\ref{rr-remove-vertices-not-enough-compatible} on the input instance\label{t2c-line-rr}\\
\If {all~$v \in V$ are pairwise compatible,}{
return~$V$
}
\Else {
Find an unmarked vertex $u$ that is incompatible to at least one other vertex\\
\texttt{$S_1 \coloneqq$ MarkedBranching($G-u,\ell,M$)}\\
\texttt{$S_2 \coloneqq$ MarkedBranching($G,\ell, M \cup \{u\}$)}\\
return~$\text{arg\,max}_{S \in \{S_1,S_2\}} |S|$
}
\caption{\texttt{MarkedBranching} Algorithm.} \label{algo-marked-branching}
\end{algorithm}

\subsection{Marking Rules} 
We describe two further rules to identify vertices that need to be marked based on the set of currently marked vertices. 

Let~$v$ and~$u$ be vertices of the input graph. If~$v$ is contained in less than $\ell$~triangles after deleting~$u$, then~$u$ has to be contained in every solution that contains~$v$. Consequently, if~$v$ is marked, we may also mark~$u$. Since a solution of~\textsc{Edge Triangle 2-Club} also fulfills the vertex-$\ell$-triangle property, this also holds for the edge-variant. 
This observation leads to the following reduction rule.

\begin{rrule}[Cascading Rule (CR)]
\label{rr-no-choice1}
Let~$v$ be a marked vertex and let~$u$ be an unmarked vertex. Let~$x_v$ be the number of triangles which contain~$v$, and let~$x_{uv}$ be the number of triangles which contain~$u$ and~$v$.
If~$x_v-x_{uv}<\ell$, then mark~$u$.
\end{rrule}

Suppose that two non-adjacent marked vertices~$u$ and~$w$ have exactly one common neighbor. Then, this unique common neighbor needs to be part of the solution~$S$ since otherwise the distance between~$u$ and~$w$ in~$G[S]$ (in case of \textsc{Vertex Triangle 2-Club}) or in~$(S,\widehat{E})$ (in case of \textsc{Edge Triangle 2-Club}) becomes larger than~$2$. Consequently, we may mark the unique common neighbor.

\begin{rrule}[No-Choice Rule (NCR)]
\label{rr-no-choice2}
Let~$u$ and~$w$ be two marked vertices. If~$uw\notin E(G)$ and both vertices have exactly one common neighbor~$v$, then mark~$v$.
\end{rrule}

\subsection{Conflict Graphs}
\label{sec-conflict-graph}

In the implementation of our algorithm, for the computation of one of the lower bounds, and for one reduction rule, we make use of the following representation of the incompatible vertex pairs of the graph: Given an instance~$(G=(V,E), \ell)$, the \emph{conflict graph of~$G$} is a graph~$G_c$ on the same vertex set as~$G$ such that two vertices are adjacent in~$G_c$ if and only if they are incompatible in~$G$.

We now bound the time needed for constructing and updating the conflict graph, as we will make use of this bound for proving running time bounds for the computation of one of the lower bounds. 

\begin{lemma}\label{lem:gc-time}
  The conflict graph~$G_c$ of a graph~$G$ can be constructed in~$\Oh(nm)$ time. 
  After deleting a vertex~$w\in V$ in~$G$, the conflict graph can be updated in $\Oh(\deg(w)\cdot m)$~time; after deleting an edge in~$G$, the conflict graph can be updated in~$\Oh(m)$ time.    
\end{lemma}

\begin{proof}
We may assume that the input graph is reduced with respect to the LDR. Therefore,~$G$ contains no isolated vertices and~$n= \Oh(m)$.

The conflict graph can be constructed as follows: We start with an empty graph. For each vertex~$v$ we perform a truncated BFS to compute~$N_2[v]$ and add conflict edges~$vw$ to~$G_c$ for every~$w \not \in N_2[v]$. 
Consequently,~$G_c$ can be computed in~$\mathcal{O}(nm)$~time.

We next consider the running time for updating~$G_c$ after deleting some vertex~$w\in V$.
To this end, note that deleting~$w$ does not repair any incompatibility, but instead new incompatibilities may emerge, since two neighbors~$x$ and~$y$ of~$w$ may now be incompatible.
To find these new incompatibilities, we perform a BFS from each neighbor of~$w$. This can be done in~$\Oh(\deg(w) \cdot m)$~time.

In case of deleting an edge~$xy$, all new arising incompatibilities contain one of the vertices~$x$ or~$y$. Thus, two calls of BFS are sufficient to update~$G_c$ accordingly in~$\Oh(m)$~time.
\end{proof}


\section{Lower Bounds}
\label{sec-lb}

To obtain good initial solutions, we implemented two heuristics.
These heuristics provide a lower bound for the size of the largest triangle~$2$-club during branching and also for some reduction rules.
The first heuristic determines the largest triangle~$2$-club in the closed neighborhood of any vertex~$v\in V(G)$ and the second one greedily determines a triangle~$2$-club in~$N_2[v]$ for any~$v\in V(G)$.

\subsection{A Lower Bound Based on Closed Neighborhoods} 
We next describe the known \emph{Neighborhood-Lower Bound (N-LB)}~\cite{AB19} and bound the worst-case running time for computing the lower bound. Moreover, we adapt the lower bound to the edge-triangle property.

For each vertex~$v\in V(G)$, we compute a largest triangle~$2$-club~$S_v$ that is contained in~$N[v]$ and also contains~$v$. 
To do this, we first construct the induced subgraph of~$N[v]$, denoted~$G_v$, and then exhaustively apply the LTR (Rule~\ref{rr-remove-vertices-to-low}) to~$G_v$. We let~$S_v$ be the vertex set of of the resulting graph. 
Note that~$S_v$ is a triangle~$2$-club:~$S_v$ fulfills the triangle property due to the LTR and~$S_v$ is a 2-club since~$v$ is universal in~$S_v$.
Moreover, since we search for a solution containing~$v$ we may abort if the LTR deletes~$v$.
The heuristic then returns a vertex set~$S$ which has maximum size of any of the sets~$S_v$, $v\in V(G)$, and the value of the N-LB is the size of~$S$. For the edge-variant, we use the same algorithm, now with the edge-variant of the LTR.

Next, we show that if a maximal triangle~$2$-club~$S$ has a universal vertex, then this heuristic will find~$S$.
In other words, in such a scenario the value of the N-LB is the size of the optimal solution.
For this, it is sufficient to show the following statement.

\begin{proposition}
\label{prop:neighborhood-heuristic-optimal}
For each~$v\in V(G)$, the solution~$S_v$ computed as described above is a largest triangle 2-club that is contained in~$G_v$.
\end{proposition}
\begin{proof}
Let~$S$ be a largest triangle~$2$-club in~$G_v$.
We show that~$S_v$ contains all vertices of~$S$.

First, we show the statement for the vertex-variant. If~$S_v$ is nonempty, then the vertex~$v$ is not returned by the LTR: $v$ is universal and thus contained in a maximum number of triangles in~$G_v$. Thus,~$S_v$ is a triangle 2-club. Now, assume towards a contradiction that~$S_v$ does not contain all vertices of~$S$. Let~$w$ be the first
vertex of~$S$ which gets deleted by the algorithm and consider the application of the
LTR that deleted~$w$.  Since~$w\in S$,
we know that~$w$ is contained in at least~$\ell$ triangles in~$G_v$, all of these
triangles are present when~$w$ is deleted, a contradiction to the definition of the LTR. Thus, such a vertex~$w$ does not exist and the algorithm returns~$S$.

Second, we show the statement for the edge-variant. Let~$G^*=(S_v,E^*)$ denote the graph at the end of the algorithm. We first show that~$S_v$ is an edge-$\ell$-triangle~$2$-club in~$G$. We show this by showing that~$G^*$ has diameter at most~2 and every edge of~$E^*$ is in at least~$\ell$ triangles in~$G^*$. The latter claim follows from the fact that the LTR has been applied exhaustively. It remains to show that~$G^*$ has diameter at most~2. To show this, we prove that in~$G^*$, the vertex~$v$ is adjacent to all other vertices. Assume that the algorithm deletes an edge~$\{v,u\}$ that is incident with~$v$. Then all other edges~$\{u,w\}$ incident with~$u$ will also be deleted by the algorithm: Since~$v$ is a universal vertex in~$G_v$, for every triangle~$\{u,w,x\}$ with~$x\neq v$, the set~$\{u,v,x\}$ is also triangle. Hence, the number of triangles containing~$\{v,u\}$ is at least as large as the number of triangles containing~$\{u,w\}$. 
As a consequence, whenever an edge~$\{v,u\}$ is deleted, all edges incident with~$u$ are deleted. 
Consequently,~$u$ is deleted as well. Altogether,~$v$ is adjacent to every other vertex of~$G^*$. Thus,~$S_v$ is an edge-$\ell$-triangle~$2$-club.

It remains to show that~$S_v$ contains all vertices of~$S$. Assume towards a contradiction that this is not the case and let~$(S,\widehat{E})$ denote the corresponding spanning subgraph. Now, assume towards a contradiction, that the
algorithm deletes at least one edge of~$\widehat{E}$. Let~$e$ be the first edge
of~$\widehat{E}$ which gets deleted, and consider the application of the LTR that deletes~$e$.  Since~$e\in \widehat{E}$ we know
that~$e$ is contained in at least $\ell$~triangles in~$(S,\widehat{E})$. Since all edges
of~$\widehat{E}$ are present when~$e$ is deleted, it is contained in at least
$\ell$~triangles at the time of its deletion, a contradiction to the definition of the LTR.  Hence, such an edge $e$ does not exist
and thus~$G^*$ contains every edge of~$\widehat{E}$ and therefore also all vertices of~$S$.     
\end{proof}

Next, we bound the overall running time of computing the N-LB, again in terms of the number of edges~$m$ and the degeneracy~$d$ of the input graph.
\begin{proposition}
\label{prop:running-time-neighborhood-heuristic}
For \textsc{Vertex Triangle~$2$-Club} and for \textsc{Edge Triangle~$2$-Club} the overall running time of the N-LB is~$\Oh(m\cdot  d^2)$.
\end{proposition}

\begin{proof}
  Fix some degeneracy ordering~$\sigma$ of~$G$.  For each vertex~$v$, the graph~$G_v$ can
  be constructed in $\Oh(\deg(v)\cdot d)$~time by considering all neighbors~$w$ of~$v$ and
  then traversing the list containing all neighbors~$u$ of~$w$ that appear after~$w$
  in~$\sigma$. Thus, the total running time for computing all induced subgraphs
  is~$\Oh(\sum_{v\in V(G)} \deg(v) \cdot d)=\Oh(m\cdot d)$. For each vertex~$v$, the
  graph~$G_v$ is~$d$-degenerate and thus has~$\Oh(\deg(v)\cdot d)$ edges. By
  Lemma~\ref{lem:ltr-time}, the exhaustive application of the LTR on~$G_v$ thus
  takes~$\Oh(\deg(v)\cdot d^2)$ time.  Hence, the overall running time for the application
  of the rule for all graphs~$v\in V(G)$
  is~$\Oh(\sum_{v\in V(G)} \deg(v) \cdot d^2)=\Oh(m\cdot d^2)$.  
\end{proof}
In the worst case, when~$d=\Theta(n)$, 
the running time bound of the N-LB becomes $\Oh(n^4)$ and thus is impractical. 
In sparse real-world graphs, however, the degeneracy takes on very small values and the running time bound guarantees that the algorithm is fast. In particular, when the degeneracy is constant, we achieve a linear running time.

\subsection{A Greedy Lower Bound}
We also use the following greedy algorithm to compute a lower bound, called \emph{Greedy-Lower Bound (G-LB)}.
For each vertex~$v\in V(G)$, we compute a triangle~$2$-club~$S_v$ that is contained in the closed second neighborhood~$N_2[v]$ and contains~$v$ as follows: 
 We first construct the induced subgraph~$G_v\coloneqq G[N_2[v]]$ and apply the LTR to~$G_v$.
By~$S_v$ we denote the vertex set of~$G_v$ after the application of the LTR.
Each vertex in~$S_v$ or each edge with two endpoints in~$S_v$ is contained in at least~$\ell$ triangles.
However,~$S_v$ is not necessarily a~$2$-club.
To solve this issue, we  test whether~$S_v$ is a~$2$-club and return a pair~$\{u,w\}$ of incompatible vertices if this is not the case.
If~$u=v$ or~$w=v$, then we delete the other vertex since we aim to find the largest triangle~$2$-club containing~$v$.
Otherwise, we greedily delete the vertex of~$u$ and~$w$ which is in less triangles in~$G_v$.
Afterwards, we again apply the LTR.
We continue with this procedure until the resulting vertex set~$S_v$ is a triangle~$2$-club. The heuristic then returns a vertex set~$S$ which has maximum size of any of the sets~$S_v$, $v\in V(G)$, and the value of the G-LB is the size of~$S$.

Next, we bound the overall running time of the G-LB.
\begin{proposition}
\label{prop:running-time-2neighborhood-heuristic}
For \textsc{Vertex Triangle~$2$-Club} and for \textsc{Edge Triangle~$2$-Club} the overall running time of the G-LB is~$\Oh(m\cdot \Delta^3\cdot d^2)$.
\end{proposition}

\begin{proof}
The graph~$G_v=G[N_2[v]]$ has less than $\deg(v)\cdot \Delta$~vertices and can be constructed in $\Oh(\deg(v)\cdot \Delta\cdot d)$~time using the degeneracy-ordering adjacency lists as detailed in the running time bound proof for the N-LB. 
  The number of triangles in~$G_v$ is~$\Oh(\deg(v)\cdot \Delta\cdot d^2)$ and again, these triangles can be computed in the corresponding running time. The total time for checking the triangle counters and deleting vertices or edges when their counters are too low is again upper-bounded by the number of triangles in~$G_v$ and thus bounded by~$\Oh\left(\deg(v)\cdot \Delta\cdot d^2\right)$. 
  The main running time bottleneck is the repeated check whether the current graph~$G_v$ has diameter~2 for which we make use of the conflict graph. By Lemma~\ref{lem:gc-time},  the conflict graph can be constructed in~$\Oh(\deg(v)^2\cdot \Delta^2 d)$ time.  
  Afterwards, we need to update the conflict graph after vertex and edge deletions. The running time for the update after deleting a vertex~$u$ is~$\Oh(\deg(u)\cdot m)$, again by Lemma~\ref{lem:gc-time}.
  The worst-case overall running for the updates thus is~$$\Oh\left(\sum_{u\in N_2[v]}\deg_v(u) \cdot \deg(v)\cdot \Delta\cdot d)=\Oh(\deg(v)^2\cdot \Delta^2\cdot d^2\right),$$ where~$\deg_v(u)$ is the degree of~$u$ in~$G_v$, if we essentially delete all vertices and edges of~$G_v$ either by reduction rules or by greedy choice. 
  The total running time for computing the lower bounds for all vertices is~$\Oh(\sum_{v\in V(G)}\deg(v)^2\cdot \Delta^2\cdot d^2)=\Oh(m\cdot \Delta^3\cdot d^2)$.
\end{proof}

Again, on very dense graphs, the running time is impractical. While the running time becomes acceptable on sparse graphs with moderate maximum degree, it is still much higher than the running time for the N-LB. 
We also observed this in our implementation, which prompted us to add several speed-ups to the implementation of the G-LB (see Section~\ref{sec:imp-detail}).

\subsection{Lower-Bound-Based Reduction Rules}
\label{sec:lb-rules}

We can make use of the lower bounds not only for pruning during the branching but also in some reduction rules. These rules are applied before and during the branching algorithm. To describe the rules, we let~$k$ denote the size~$k$ of a largest triangle~$2$-club detected so far.
In other words, whenever we apply these rules, the aim is to find a triangle~$2$-club of size at least~$k+1$. Now, if a vertex~$u$ is compatible with at most~$k-1$ other vertices, then~$u$ cannot be contained in a solution of size larger than~$k$. If~$v$ is unmarked, then~$v$ can be deleted from the input instance. 
Otherwise, if~$v$ is marked, we may discard the instance.

We present two variants of this rule. In the 2-NR (Rule~\ref{rr-remove-vertices-too-small-2-neighborhood}) we directly count the number of compatible vertices of some vertex~$v$, that is, we check if~$N_2[v]$ has size at least~$k$. In the LCR (Rule~\ref{rr-remove-vertices-not-enough-compatible}) we make use of the conflict graph.


\begin{rrule}[$2$-Neighborhood Rule (2-NR)]
\label{rr-remove-vertices-too-small-2-neighborhood}
Delete all vertices~$v$ such that $|N_2[v]|\le k$. If a marked vertex is deleted, report that there is no solution of size larger than~$k$.
\end{rrule}


\begin{rrule}[Low-Compatibility Rule (LCR)]
\label{rr-remove-vertices-not-enough-compatible}
Delete vertices which have a degree of at least~$n-k+1$ in the conflict graph.
If a marked vertex is deleted, report that there is no solution of size larger than~$k$.
\end{rrule}


Finally, we can make use of the following reduction rule which computes an upper bound based on the conflict graph. This rule was also used for other \textsc{2-Club} variants~\cite{HKN15,KNNP19}.

\begin{rrule}[Matching Rule]
\label{rr-matching-in-conflict-graph}
Compute the size~$b$ of a maximum matching for~$G_c$.
If~$|V(G_c)|-b$ is at most~$k$, then return no.
\end{rrule}
\begin{proposition} Rule~\ref{rr-matching-in-conflict-graph} is correct.
\end{proposition}

\begin{proof}
Let~$M$ be a matching of size~$b$ in~$G_c$. To prove the correctness of the rule, we show that every solution~$S$ has size at most~$|V(G_c)|-b$. 

Let~$S$ be a solution. Then,~$S$ does not contain two vertices that are adjacent in~$G_c$ and thus, $S$~contains at most~$b$ vertices that are endpoints of edges in~$M$. Therefore,~$S$ consists of at most~$|V(G_c)|-b$ vertices.
\end{proof}

\section{Implementation Details}
\label{sec:imp-detail}

\begin{figure}
\centering
\begin{tikzpicture}
\draw[draw=black] (-0.9,-0.9) rectangle ++(1.8,0.6);
\node[label={Basic Rules}](A1) at (0, -0.95) {};

\draw[draw=black,fill=light-gray] (1.7,-1.8) rectangle ++(1.8,2.3);
\node(Z1) at (0.9, -0.6) {};
\node(Z2) at (1.7, -0.6) {};
\path [->,line width=0.5mm] (Z1) edge (Z2);

\node[label={N-LB}](B1) at (2.6, -0.25) {};
\node(E1) at (2.6, -0.05) {};
\node(E2) at (2.6, -0.75) {};
\path [->,line width=0.5mm] (E1) edge (E2);
\node[label={2-NR}](B2) at (2.6, -1.25) {};
\node[label={Basic Rules}](B3) at (2.6, -1.75) {};

\draw[draw=black,fill=lightgray] (4.3,-1.8) rectangle ++(1.8,2.3);
\node(X1) at (3.5, -0.6) {};
\node(X2) at (4.3, -0.6) {};
\path [->,line width=0.5mm] (X1) edge (X2);

\node[label={G-LB}](C1) at (5.2, -0.25) {};
\node(F1) at (5.2, -0.05) {};
\node(F2) at (5.2, -0.75) {};
\path [->,line width=0.5mm] (F1) edge (F2);
\node[label={2-NR}](C2) at (5.2, -1.25) {};
\node[label={Basic Rules}](C3) at (5.2, -1.75) {};

\draw[draw=black] (6.8,-2.35) rectangle ++(4.4,3.7);
\node(X1) at (6.1, -0.6) {};
\node(X2) at (6.9, -0.6) {};
\path [->,line width=0.5mm] (X1) edge (X2);

\node[label={For each~$v\in V(G)$ do}](D0) at (9, 0.5) {};
\node[label={\texttt{MarkedBranching}$(G,\ell,\{ v\})$}](D1) at (9, 0) {};
\node(G1) at (9, 0.2) {};
\node(G2) at (9, -0.5) {};
\path [->,line width=0.5mm] (G1) edge (G2);
\node[label={uses}](RR) at (9.5, -0.45) {};
\node[label={Basic Rules}](D2) at (9, -1) {};
\node[label={LCR/2-NR}](D3) at (9, -1.5) {};
\node[label={Marking Rules}](D4) at (9, -2) {};
\node[label={Matching Rule$^*$}](D5) at (9, -2.5) {};
\end{tikzpicture}
\caption{Sequence of our algorithm. 
Everything described in white boxes is done in each of the four variants.
Everything within the \colorbox{light-gray}{light-gray box} is done in the two variants using a lower bound (\texttt{N-LB} and \texttt{Multi-LB}).
Furthermore, the \colorbox{lightgray}{gray box} is used only for \texttt{Multi-LB}.
The Matching Rule (marked with an asterisk (*)), is used in each variant except \texttt{Basic}.}
\label{fig-algo-seqeunce}
\end{figure}
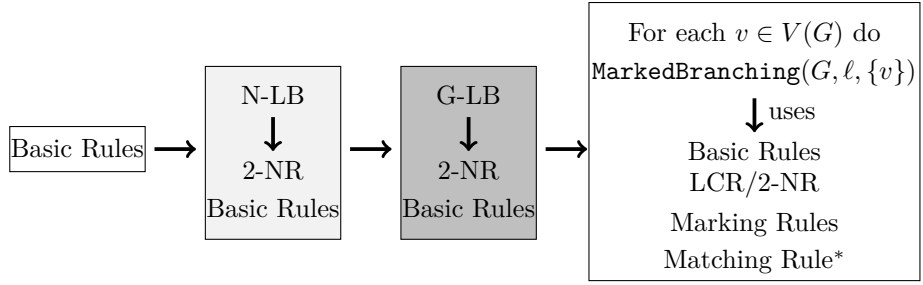

To evaluate the effectiveness of our ideas described in Sections~\ref{sec-bran} and~\ref{sec-lb} in terms of decreasing the total running time, we use four different variants of our algorithm in our evaluation.
An outline of our algorithm and our four variants (\texttt{Basic}, \texttt{Basic+UB}, \texttt{N-LB}, and \texttt{Multi-LB}) is given in Figure~\ref{fig-algo-seqeunce}. Herein, we refer to the IRR, MIR, CR, and NCR as the \emph{Marking Rules}.

\texttt{Basic} uses all reduction rules except the one based on the conflict graph, \texttt{Basic+UB} additionally uses the conflict graph, \texttt{N-LB} additionally uses the N-LB, and \texttt{Multi-LB} additionally uses the G-LB.

\paragraph{Data Structures.}
Our algorithm assigns each vertex a unique ID.
To represent a graph we use hash sets and hash maps.
More precisely, for the vertex set we rely on a hash set to allow for a fast check whether a vertex is present.
The set of edges is organized in adjacency lists. 
More precisely, a hash map uses the vertex IDs to
assign each vertex a distinct hash set, containing its neighbors, that is, the IDs of each adjacent vertex. 
Furthermore, we use a stack to reverse the operations in the search tree efficiently.

After the initial exhaustive application of the LDR before we compute the lower bounds (the N-LB or the G-LB), we compute all triangles of the graph.
Then, for each vertex~$v$ a hash map with all triangles containing~$v$ is created.
This allows for a fast check of all operations related to triangles, for example, whenever we compute the set of triangles containing one edge~$vw$.

\paragraph{Compatibility Test and Conflict Graph.}
To test whether two vertices~$u$ and~$w$ are compatible, that is, have distance at most~$2$, we first check whether~$u=w$.
Afterwards, we check whether~$u$ and~$w$ are adjacent and then we check if~$u$ and~$w$ have a common neighbor.

Before we branch on some vertex~$v$, we create the conflict graph~$G_c$ of~$G_v$ by using the above-mentioned compatibility test.
We do \emph{not} create the conflict graph of the complete input graph since this is too time-consuming and needs too much memory.
Whenever a vertex~$w$ is deleted from~$G_v$, we also delete~$w$ from~$G_c$.
Afterwards, we update~$G_c$ according to the algorithm described in the proof of Lemma~\ref{lem:gc-time}.

\paragraph{Reduction Rules.}

Next, we give some implementation details for reduction rules which are not implemented in a straightforward manner.
Let~$d_{\text{min}}$ be the minimum degree of any vertex in a triangle~$2$-club.
Recall that~$d_{\text{min}}$ is at most~$ 1/2+\sqrt{1/4+2\ell}$ for the vertex-variant and that~$d_{\text{min}}=\ell+1$ for the edge-variant.
In one application of the LDR (Rule~\ref{rr-remove-vertices-to-low-degree}), we first store all vertices of~$G$ that have degree smaller than~$d_{\text{min}}$.
Afterwards, we delete all of these vertices.
In one application of the LTR (Rule~\ref{rr-remove-vertices-to-low}) we similarly store all vertices or edges which are in less than~$\ell$ triangles, and then delete them simultaneously.
Also, if we apply the LDR or the LTR during branching, that is, when we aim to find a largest solution containing vertex~$v$, we abort this branch whenever the LDR or the LTR deletes~$v$.

In our implementations of the IRR (Rule~\ref{rr-remove-incompatible}), the MIR (Rule~\ref{rr-trivial-no}), and the LCR (Rule~\ref{rr-remove-vertices-not-enough-compatible}), which are all applied during branching on some vertex~$v$, we count the number of incompatibilities of each vertex~$w\in V(G_v)$ and stop if~$|V(G_v)|-k$ incompatibilities are found (where~$k$ is the size of a largest triangle~$2$-club detected so far).

\paragraph{LCR vs 2-NR.} 
Recall that the LCR and the 2-NR (Rule~\ref{rr-remove-vertices-too-small-2-neighborhood}) serve the same purpose; identifying vertices which have at most~$k-1$ compatible vertices.
The 2-NR is faster than the LCR if~$N_2[v]$ is small.
Hence, we have chosen to use the 2-NR for low densities.
Another reason for our choice is that if the density of~$N_2[v]$ is large,  the LCR is not much slower than the 2-NR.
Here, the \emph{density} is the ratio between the edges present in a graph and the maximum number of edges that the graph can contain.
In other words, on the basis of the density of~$N_2[v]$ the algorithm decides whether it uses the 2-NR or the LCR in each subsequent call of \texttt{MarkedBranching}$(G,\ell,\{v\})$.
We use the 2-NR if the density of~$N_2[v]$ is at most~$0.05$; otherwise the LCR is used.

\paragraph{Lower Bounds.}
For the G-LB and branching, we sort the vertices descendingly according to the size of their $2$-neighborhood.
Since the computation of the N-LB is fast, we use an arbitrary vertex ordering to compute the N-LB.
After the branching on vertex~$v$ is completed, vertex~$v$ is deleted from the graph.

Before we compute the G-LB, we use the N-LB, to obtain an initial lower bound.
The subsequent applications of the Basic Rules and the 2-NR delete many further vertices from the graph and thus the G-LB has to be computed only for a fraction of the initial vertices.
Since the running time for the G-LB is much larger than for the N-LB, one would expect that this order of applying the lower bounds decreases the total running time. 
This was confirmed in preliminary experiments.
In the following, we denote this combined lower bound by \emph{Multi-LB}.

Before we compute the G-LB for some vertex~$v$, we compute an upper bound for the size of a solution containing vertex~$v$ by checking whether the size of~$N_2[v]$ is larger than~$k$.
If this is  not the case, no solution of size at least~$k+1$ containing~$v$ exists, and thus it is safe to delete~$v$ from~$G$.
The deletion of~$v$ is beneficial since after this deletion we may now also observe that~$|N_2[w]|\le k$ for some other vertex~$w$, and hence we may also delete~$w$ from~$G$.
If~$|N_2[v]|>k$, we construct the subgraph induced by~$N_2[v]$. 
Afterwards, we use the LDR as described above to delete vertices which cannot have sufficiently many triangles.

After the computation of the N-LB, we apply the LTR in the following way:
Step 1: Apply LDR exhaustively.
Step 2: Perform one application of the LTR. 
If at least one vertex was deleted in this application, go back to Step~1.

After the computation of the Multi-LB, we construct the conflict graph (see Section~\ref{sec-conflict-graph}).
This allows us to use the LCR, that is, to delete vertices which cannot be part of a solution of at least~$k+1$ vertices.
Next, we establish the triangle property as follows:
\begin{enumerate}[label=\arabic*.]
\item We apply the LCR.
\item If the algorithm deleted at least one vertex in Step~1, we apply the LDR and  go back to Step~1.
\item We perform one iteration of the LTR, that is, we delete all vertices/edges in less than~$\ell$ triangles.
\item If the algorithm deleted at least one vertex in Step~3, we apply the LDR and  go back to Step~1.
\end{enumerate}

When the algorithm does not go back to Step~1 in Step~4, the triangle property is established.
Now, if the number of the remaining vertices in~$G_v$ after the application of all these reduction rules is at most~$k$, we delete~$v$ from the input graph.
This is correct since until this point we only applied reduction rules and hence there is no solution of size at least~$k+1$ containing~$v$.
With the deletion of~$v$ of the input graph during the lower bound computation we avoid having to apply all these reduction rules for~$v$ again during branching.

Next, similar to Rule~\ref{rr-matching-in-conflict-graph} we greedily compute a maximal matching~$M$ of the conflict graph.
For each edge in~$M$, we greedily delete the vertex which is in less triangles (vertex-variant) or has smaller degree (edge-variant).
Note that if any edge of~$M$ contains~$v$, then we always delete the other vertex because of the premise that we search for a solution containing~$v$.
Now, we repeat the above procedure of establishing the triangle property and computing a matching until the remaining vertex set fulfills the triangle and the diameter property.

Our implementation of the Multi-LB seems rather complicated, but preliminary experiments showed that this procedure is indeed necessary to achieve small running times.

\section{Experiments}

\subsection{Experimental Setup}

Each experiment was performed on a single thread of an Intel(R) Xeon(R) Silver 4116 CPU with
2.1 GHz, 24~CPUs and 128 GB RAM running Java openjdk 17.0.3.
Our algorithms are implemented in Java.
As benchmark data set we used 67 social, biological, and technical networks obtained
from the Network Repository~\cite{nr}, KONECT~\cite{Kun13}, and the 10th DIMACS
challenge~\cite{BKM+14}.
These networks range from less than 100 vertices to up to $300\,000$ vertices.
All networks except five are sparse with density at most~$0.05$.
Roughly 10 networks have less than 100 vertices, 25 networks have between 100 and~$1\,000$ vertices, 25 have between~$1\,000$ and~$10\,000$ vertices, and 10 have more than~$10\,000$ vertices.
Furthermore, we tested 27 different values for~$\ell$ in the range of 1 to 100.
More precisely, we tested our algorithm for each value in $$\{1, 2, \ldots, 6, 7, 9, 11, 13, 15, 20, 25, \ldots, 90, 100\}.$$

\begin{figure}[t!]
\centering
\begin{minipage}[b]{0.45\textwidth}
\includegraphics[width=\textwidth]{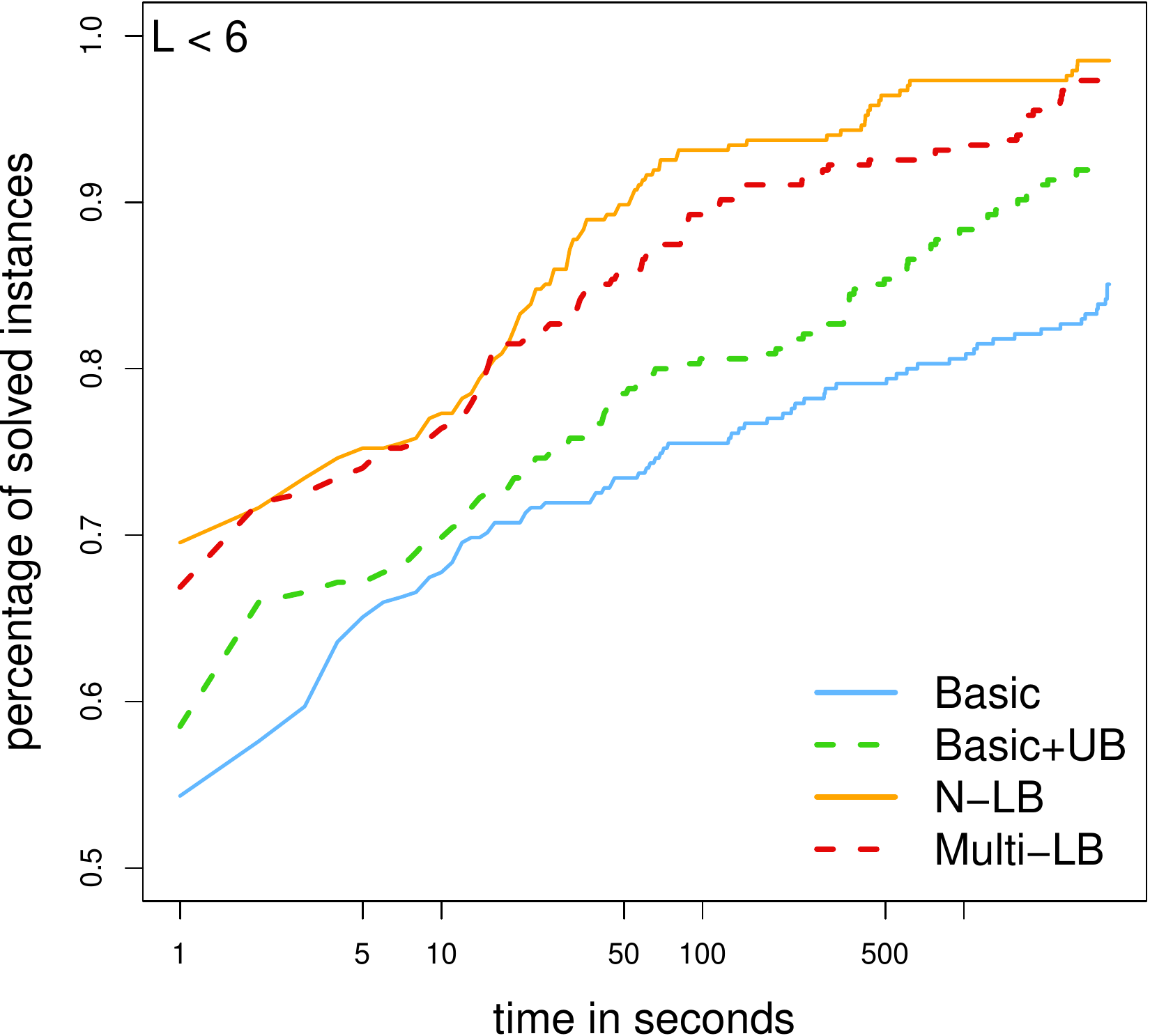}
\end{minipage}
\begin{minipage}[b]{0.45\textwidth}
\includegraphics[width=\textwidth]{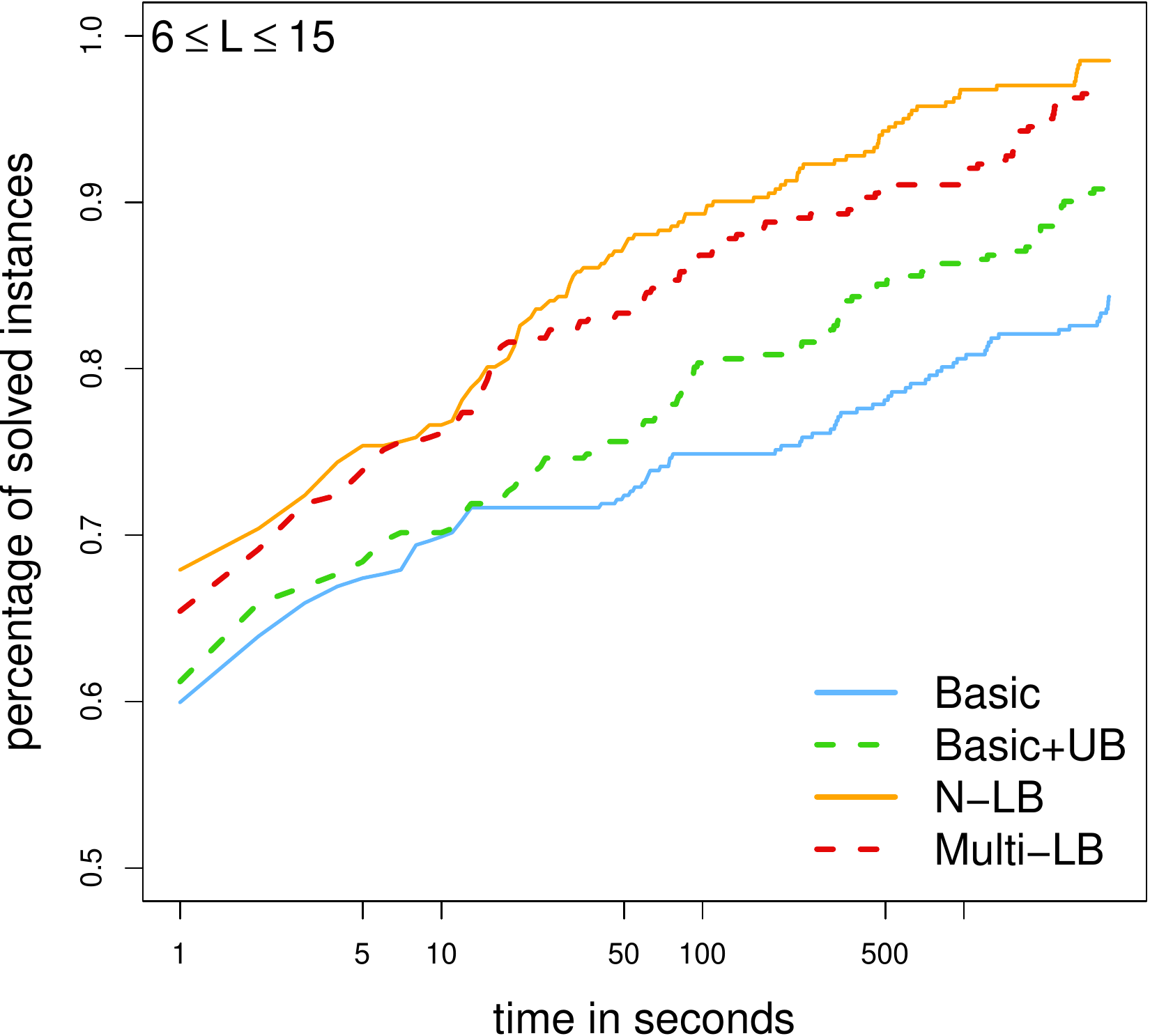}
\end{minipage}
\caption{Comparison of the four variants of our algorithm for \textsc{Vertex Triangle~$2$-Club} for~$\ell\le 15$.
}
\label{fig-results-vertex-variant}
\end{figure}

For each instance, we set a time-out of 1~hour.
The time needed to read the graph is not included in the running time.
Our source code, the list of all networks used in our experiments, and our result files are available at \url{https://www.uni-marburg.de/en/fb12/research-groups/algorith/t2c.zip}.\\
Almeida and Brás~\cite{AB19} provided four different ILP formulations for \textsc{Vertex Triangle~$2$-Club}.
They implemented their ILPs in CPLEX (\url{https://www.ibm.com/de-de/products/ilog-cplex-optimization-studio}).
We run their ILP variants with CPLEX 20.1.

\subsection{Results for the Variants of our Branching Algorithm}


\begin{figure}[t!]
\centering
\begin{minipage}[b]{0.45\textwidth}
\includegraphics[width=\textwidth]{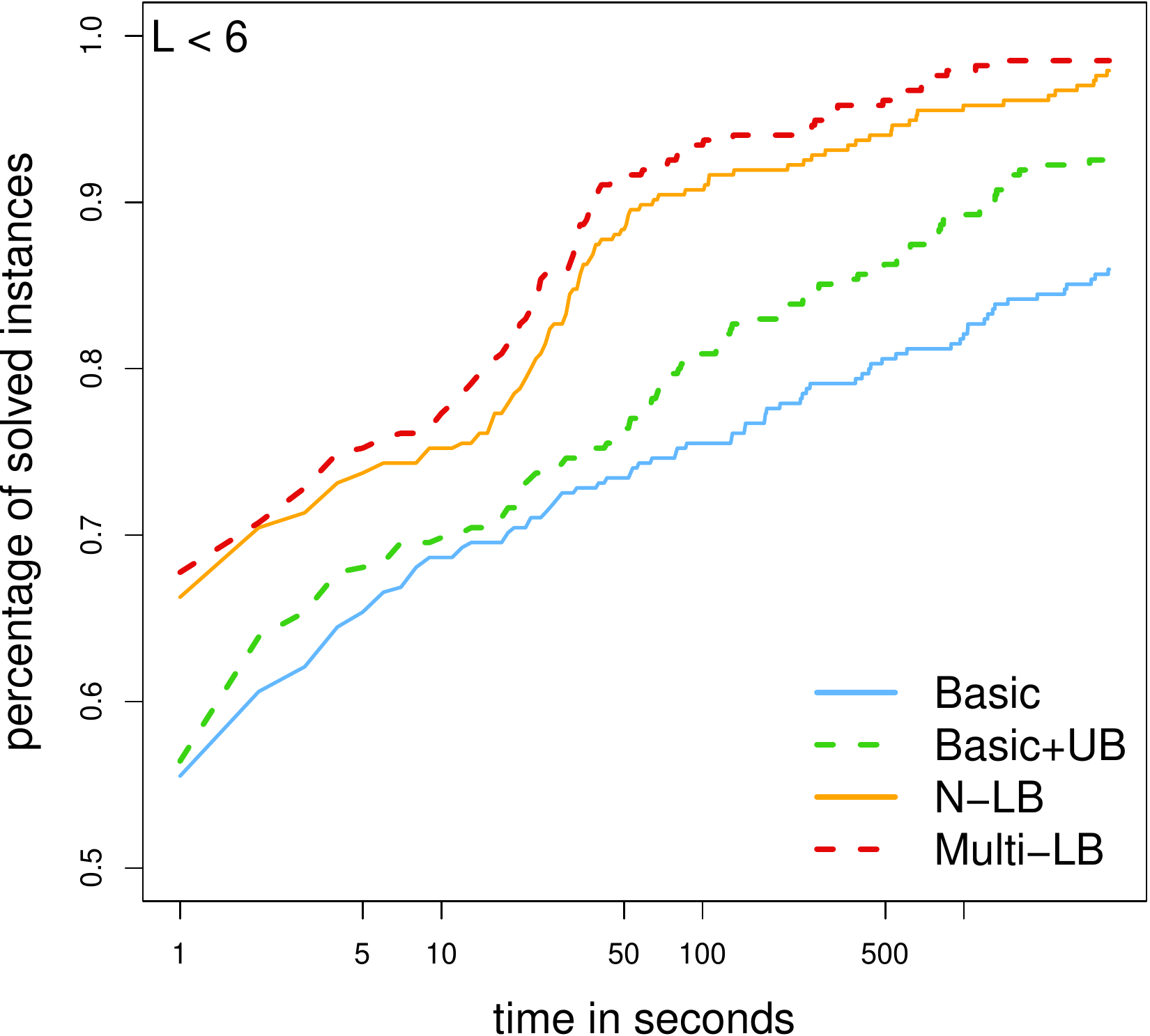}
\end{minipage}
\begin{minipage}[b]{0.45\textwidth}
\includegraphics[width=\textwidth]{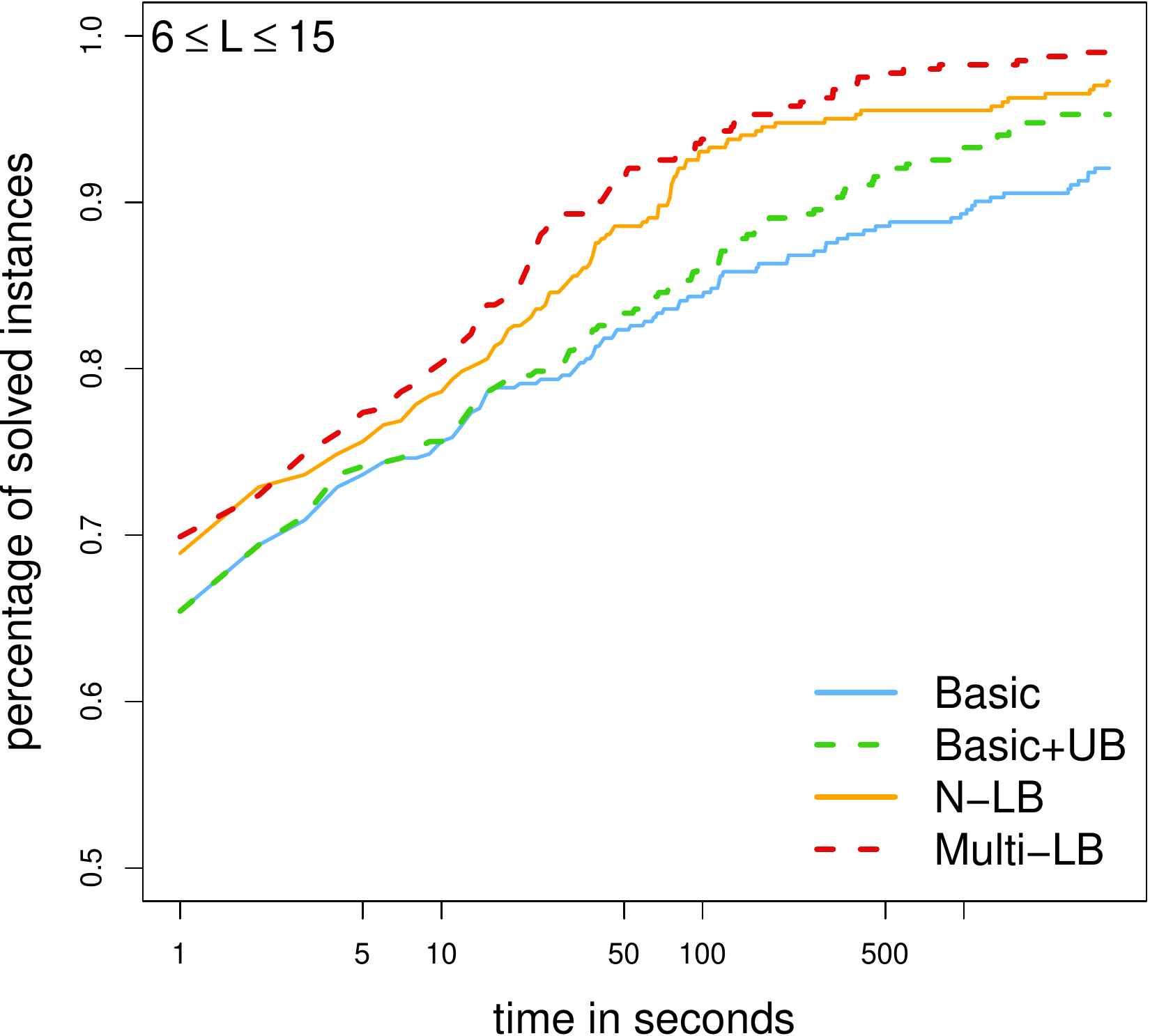}
\end{minipage}
\caption{Comparison of the four variants of our algorithm for \textsc{Edge Triangle~$2$-Club} for~$\ell\le 15$.
}
\label{fig-results-edge-variant}
\end{figure}

\begin{figure}[t!]
\centering
\begin{minipage}[b]{0.45\textwidth}
\includegraphics[width=\textwidth]{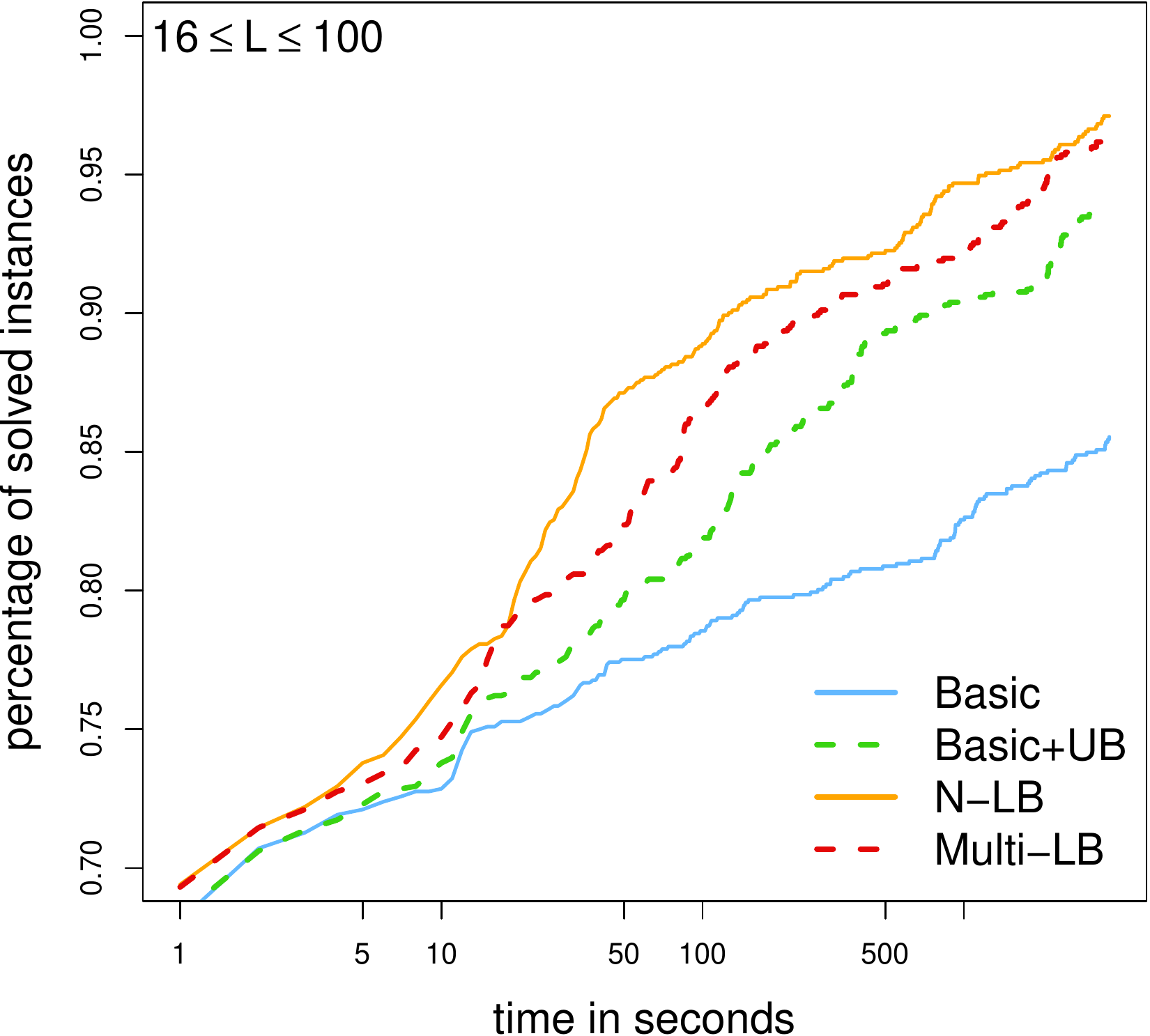}
\end{minipage}
\begin{minipage}[b]{0.45\textwidth}
\includegraphics[width=\textwidth]{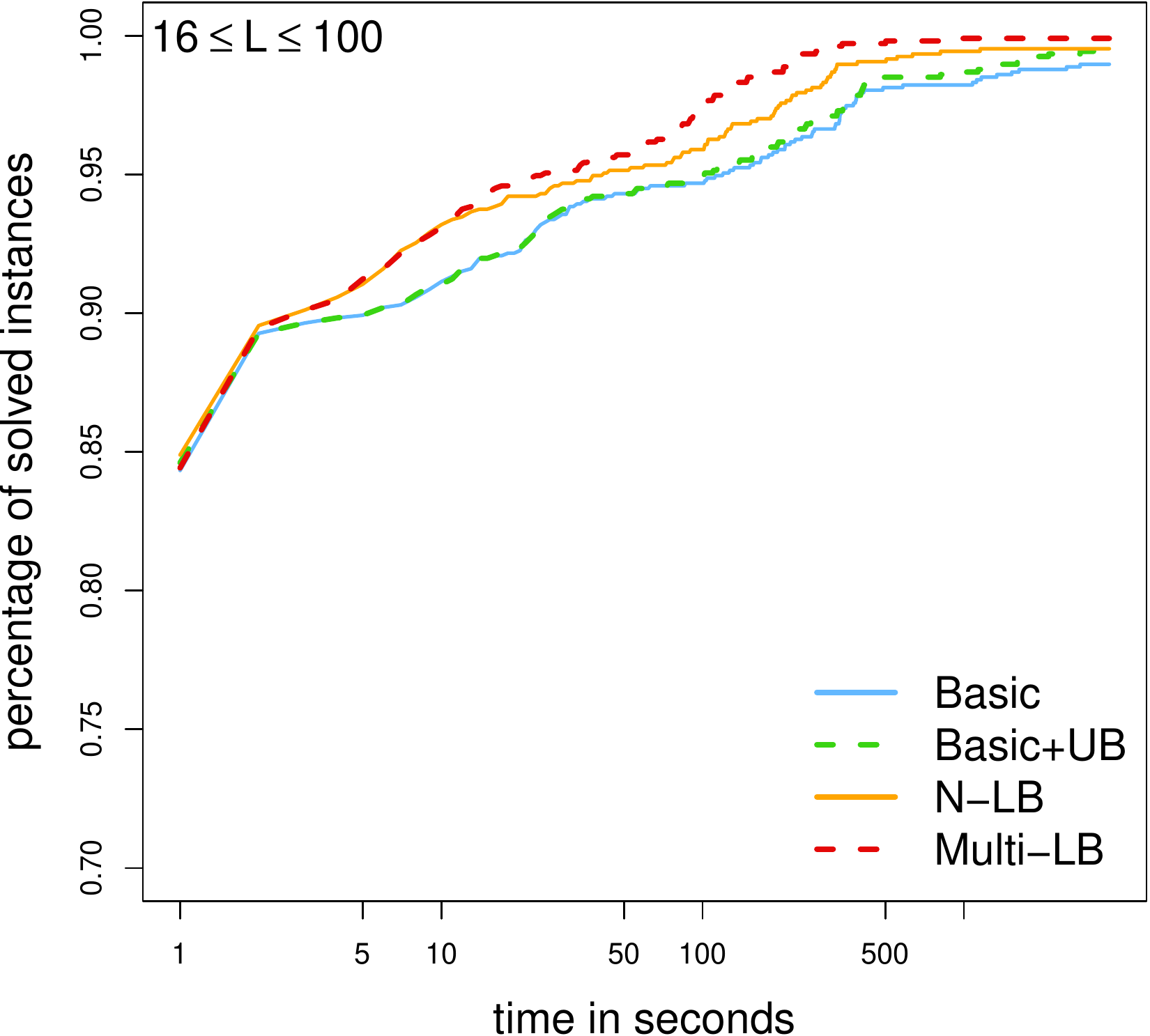}
\end{minipage}
\caption{Comparison of the four variants of our algorithm for~$\ell\ge 16$. 
The left plot shows our results for \textsc{Vertex Triangle~$2$-Club} and the right plot shows our results for \textsc{Edge Triangle~$2$-Club}.}
\label{fig-results-large-ell}
\end{figure}

The performance of the four variants of our algorithm for \textsc{Vertex Triangle~$2$-Club} and~$\ell\le 5$ is shown in the left part of Figure~\ref{fig-results-vertex-variant}.
\texttt{Basic} is substantially slower than \texttt{Basic+UB} which in turn is substantially slower than \texttt{Multi-LB}.
Furthermore, \texttt{N-LB} is even faster than \texttt{Multi-LB}.
The performance of the four variants of our algorithm for \textsc{Vertex Triangle~$2$-Club} and~$6\le \ell\le 15$ is shown in the right part of Figure~\ref{fig-results-vertex-variant}.
All four variants are slightly faster for~$6\le \ell\le 15$, where the LDR (Rule~\ref{rr-remove-vertices-to-low-degree}) and the LTR (Rule~\ref{rr-remove-vertices-to-low}) are applied more often in the initial data reduction which deletes substantially more vertices before the algorithm computes a lower bound and applies branching.


The performance of the four variants of our algorithm for \textsc{Edge Triangle~$2$-Club} is shown in Figure~\ref{fig-results-edge-variant}.
The left part shows our results for~$\ell\le 5$ and the right part shows our results for~$6\le \ell\le 15$.
As in the vertex-variant, \texttt{Basic} is substantially slower than \texttt{Basic+UB} which in turn is substantially slower than \texttt{N-LB}.
In contrast to the vertex-variant, \texttt{Multi-LB} is faster than \texttt{N-LB}.
In the edge-variant, all variants are substantially faster for larger~$\ell$.

The performance of the four variants for~$\ell> 15$ of our algorithm for \textsc{Vertex Triangle~$2$-Club} and \textsc{Edge Triangle~$2$-Club} is shown in Figure~\ref{fig-results-large-ell}.
For \textsc{Vertex Triangle~$2$-Club} the results are similar to~$\ell\le 15$, that is, \texttt{Basic} is substantially slower than \texttt{Basic+UB} which in turn is substantially slower than \texttt{Multi-LB}, and \texttt{N-LB} is the fastest.
For \textsc{Edge Triangle~$2$-Club} the result is different: \texttt{Basic} is faster than \texttt{Basic+UB} which is slower than \texttt{N-LB} and \texttt{Multi-LB} is the fastest.
Again, all four variants are substantially faster for larger~$\ell>15$ than for~$\ell\le 15$ for both \textsc{Vertex Triangle~$2$-Club} and \textsc{Edge Triangle~$2$-Club}.
The main reason for this result is that the LDR and the LTR are applied more often in the initial data reduction which deletes a large portion of the vertices from the graph.
This is especially the case for \textsc{Edge Triangle~$2$-Club} since initially all vertices with degree at most~$\ell$ get deleted.

\begin{table}[t!]
\caption{Average quality of both lower bounds for both the edge and the vertex-variant for different values of~$\ell$.}
{\footnotesize
\begin{tabularx}{\textwidth}{p{1.5cm} X p{2cm} X X p{2cm} X}
  \toprule
\multirow{2}{*}{LB} & \multicolumn{3}{c}{Vertex-Variant} & \multicolumn{3}{c}{Edge-Variant} \\
\cmidrule(l{0em}r{1em}){2-4}
\cmidrule(l{0em}r{1em}){5-7}

& $\ell\le 5$ & $6\le\ell\le 15$ & $\ell\ge 16$ & $\ell\le 5$ & $6\le\ell\le 15$ & $\ell\ge 16$ \\
\midrule
N-LB   & $95.8\%$ & $93.7\%$ & $94.2\%$ & $96.8\%$ & $96.7\%$ & $97.2\%$ \\
Multi-LB & $99.9\%$ & $99.8\%$ & $99.3\%$ & $99.9\%$ & $99.9\%$ & $99.9\%$ \\
\bottomrule
\end{tabularx}
}\label{tab-lb-quality}
\end{table}

\begin{figure}[t!]
\centering
\begin{minipage}[b]{0.45\textwidth}
\includegraphics[width=\textwidth]{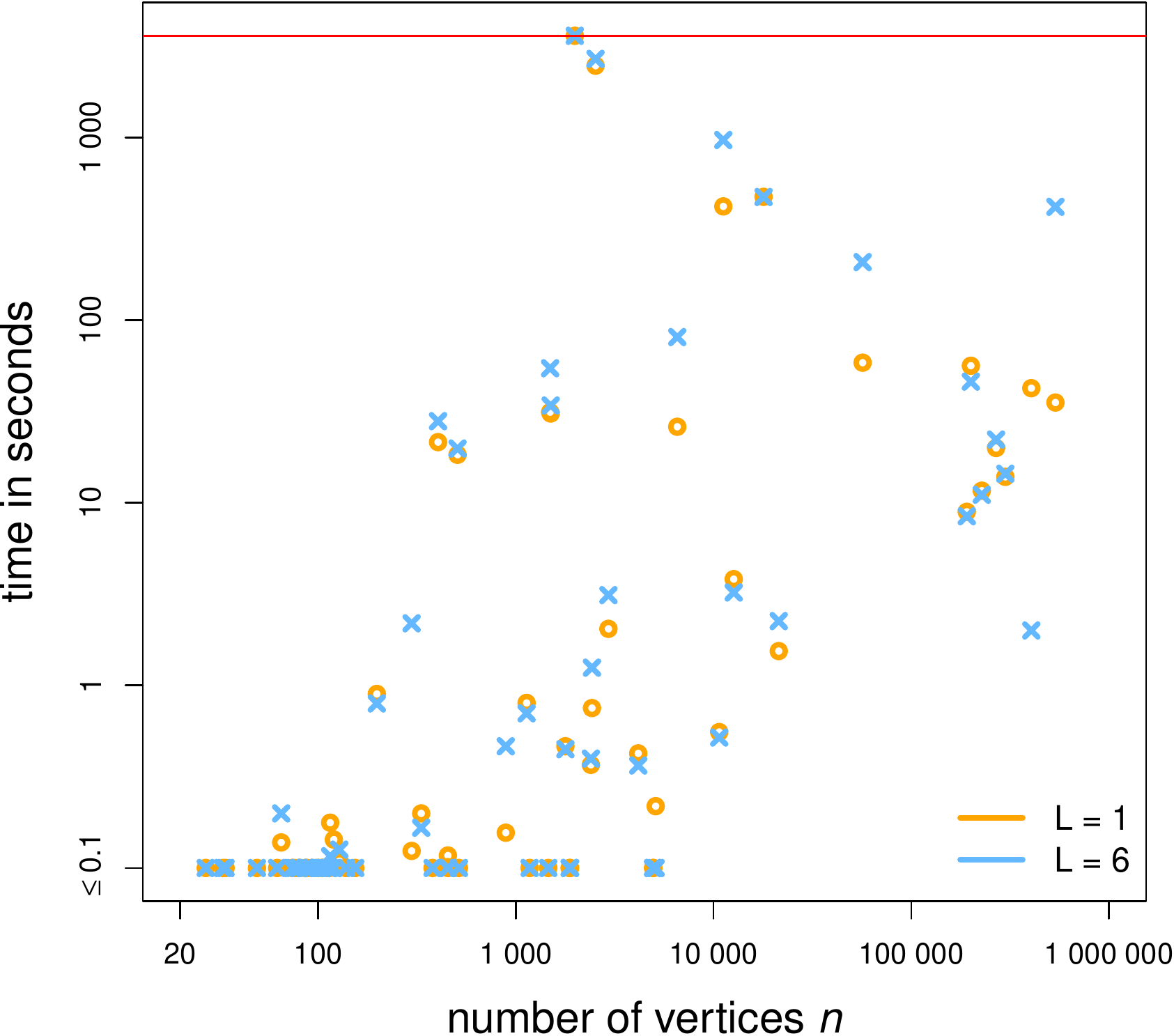}
\end{minipage}
\begin{minipage}[b]{0.45\textwidth}
\includegraphics[width=\textwidth]{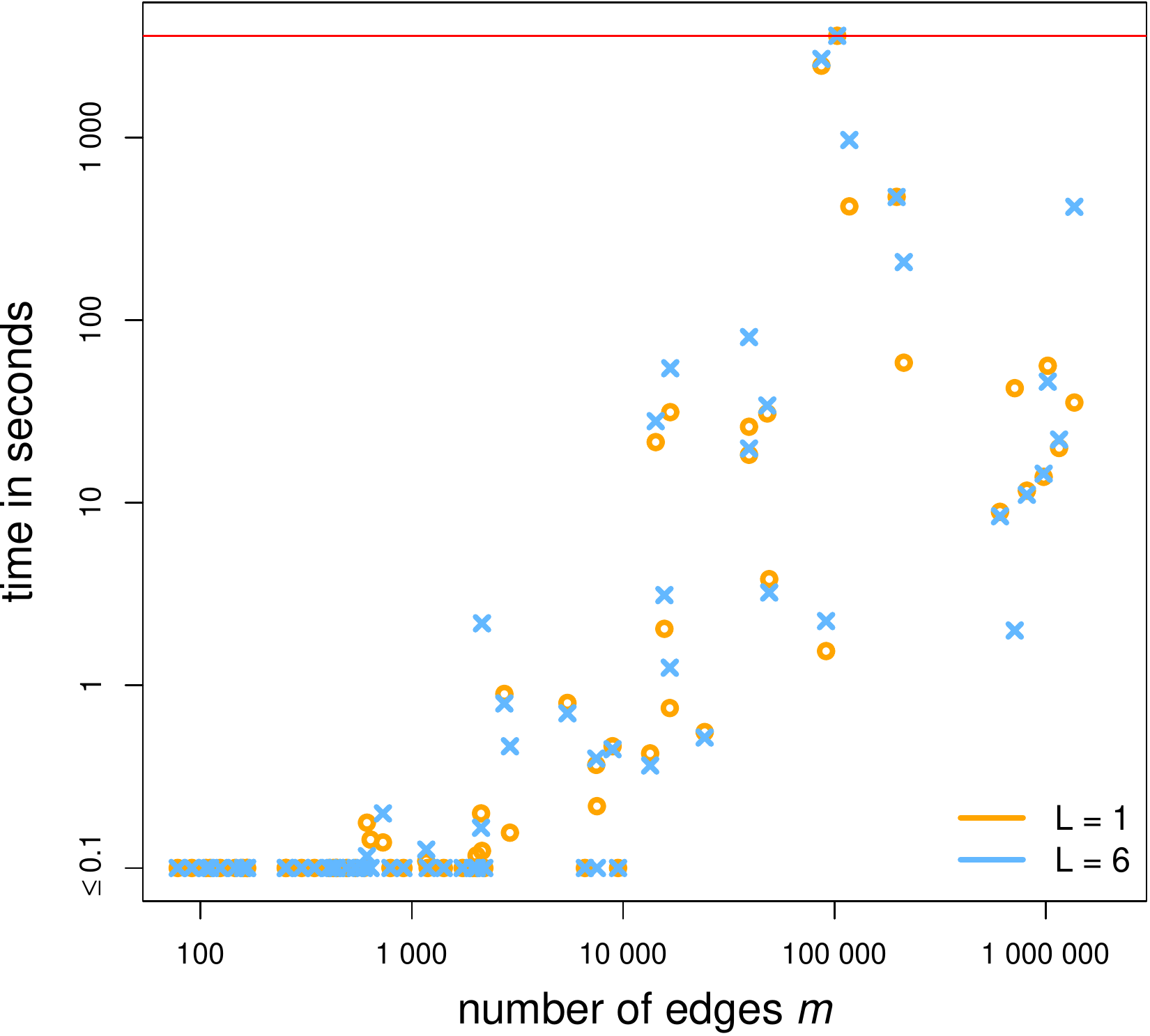}
\end{minipage}
\caption{Dependence of the running time of \texttt{N-LB} for \textsc{Vertex Triangle~$2$-Club} on the number of vertices (left) and the number of edges (right) of the input graph. Points on the red line correspond to instances which could not be solved within the time limit of 1h.}
\label{fig-results-dependece-ell-to-n-or-m}
\end{figure}

Interestingly, \texttt{Multi-LB} was only beneficial in terms of running time for the edge-variant; for the vertex-variant the running time increased compared with \texttt{N-LB}.  
However, such a  behavior cannot be observed if we compare the size of the lower bounds N-LB (used in \texttt{N-LB}) and Multi-LB (used in \texttt{Multi-LB} and denoted with Multi-LB) with the size of an optimal solution, see Table~\ref{tab-lb-quality}.
For both the edge and the vertex-variant, the Multi-LB is much better than the N-LB. 
Surprisingly, the difference of the quality of these lower bounds is larger in the vertex-variant than in the edge-variant despite the fact the \texttt{Multi-LB} is only faster than the \texttt{N-LB} for the edge-variant.

The performance of our fastest variant \texttt{N-LB} for \textsc{Vertex Triangle~$2$-Club} in terms of the number of vertices and the number of edges of the input graph for the two specific values of~$\ell=1$ and~$\ell=6$ is shown in Figure~\ref{fig-results-dependece-ell-to-n-or-m}.
The only graph for both~$\ell=1$ and~$\ell=6$ which was not solved within 1~hour is AllActors.
With larger graph size the running time increases. 
There is no clear dependence of the running time with~$n$.
For example, two instances with less than $1\,000$~vertices require more time than some instances with more than $100\,000$~vertices.
In contrast, there is a clear dependence of the running time with the number of edges~$m$.

\begin{figure}[t!]
\centering
\begin{minipage}[b]{0.45\textwidth}
\includegraphics[width=\textwidth]{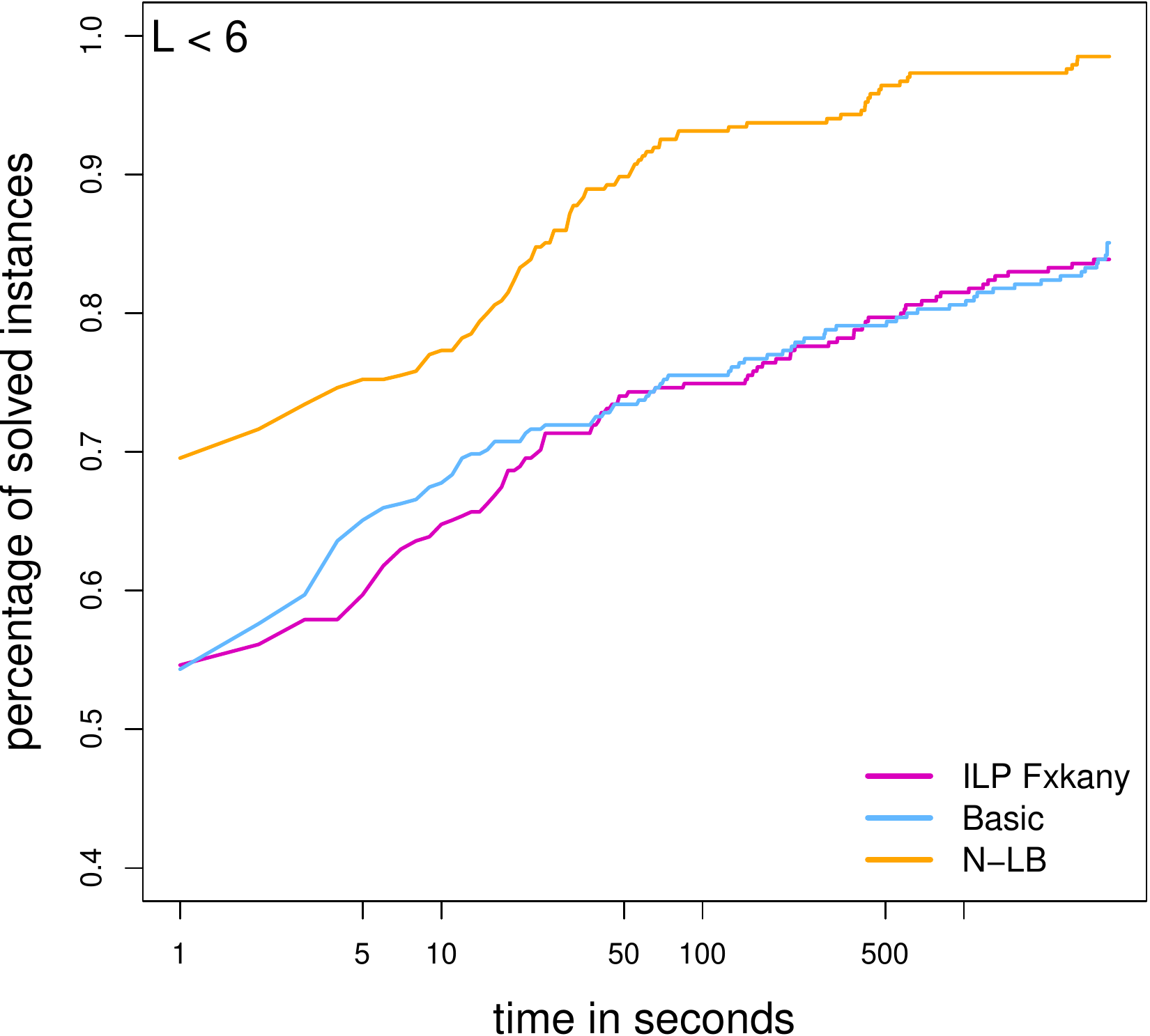}
\end{minipage}
\begin{minipage}[b]{0.45\textwidth}
\includegraphics[width=\textwidth]{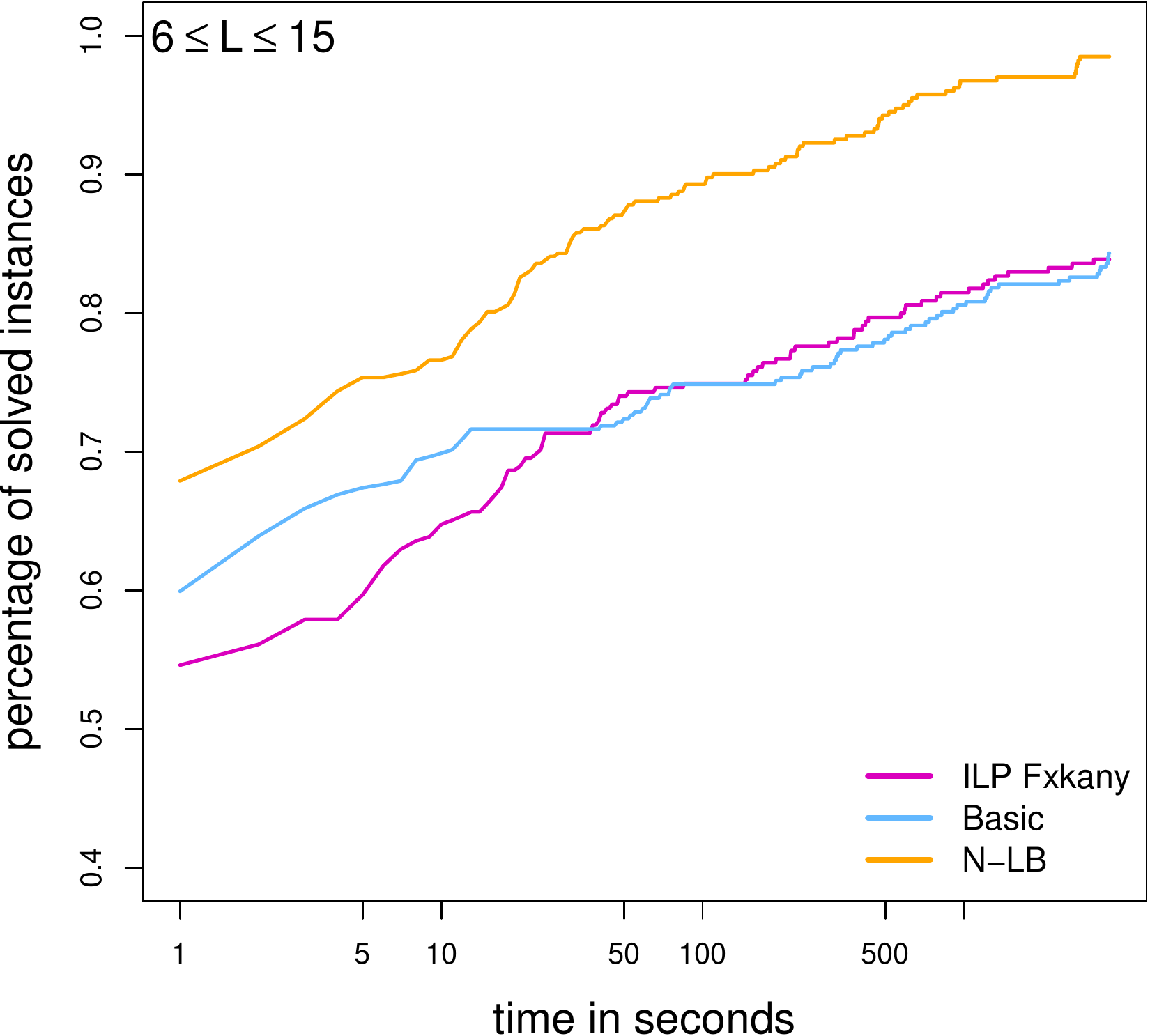}
\end{minipage}
\begin{minipage}[b]{0.45\textwidth}
\includegraphics[width=\textwidth]{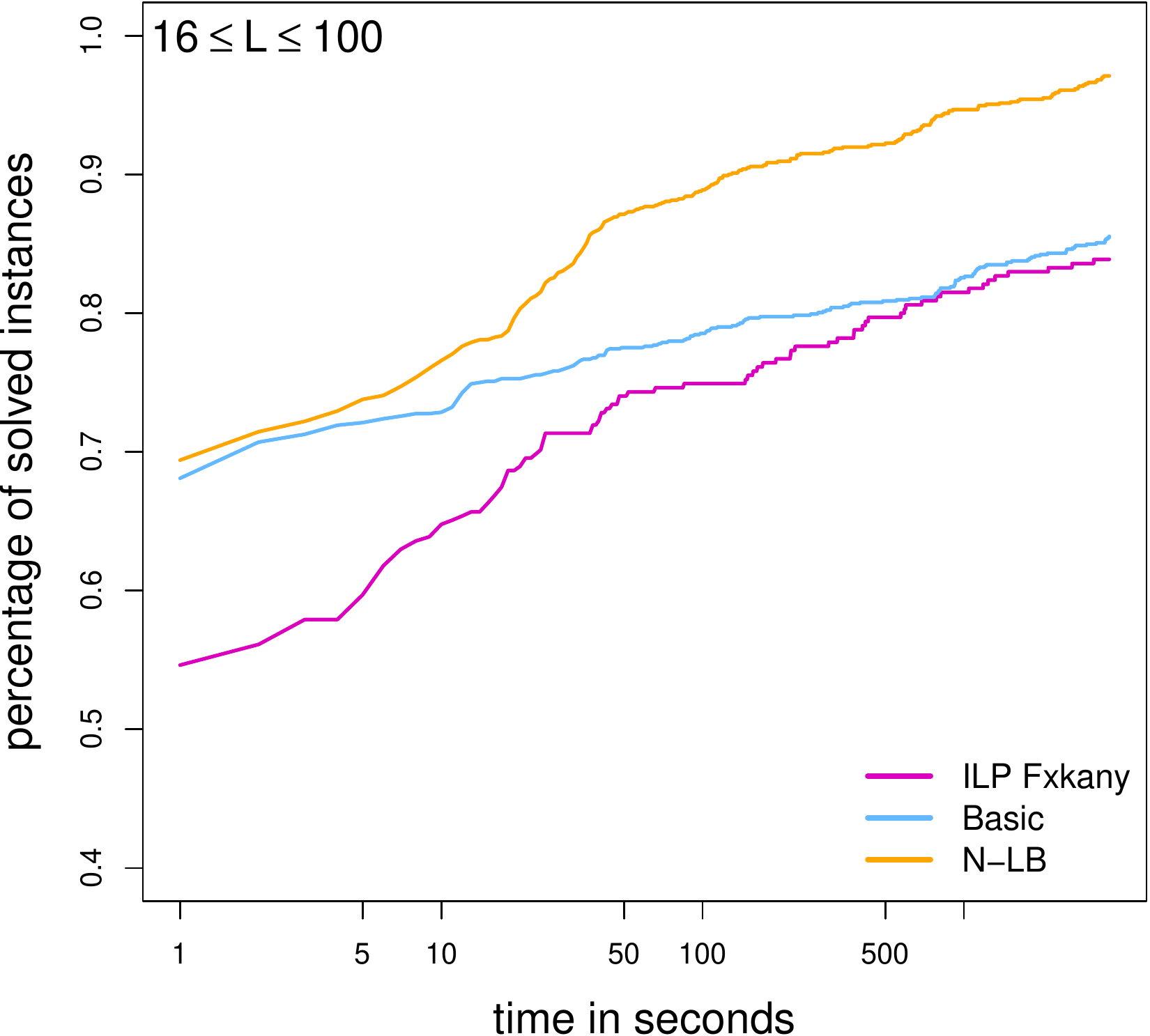}
\end{minipage}
\caption{Comparison of the 4 variants of our algorithm for \textsc{Vertex Triangle~$2$-Club} with the fastest ILP variant of Almeida and Br{\'{a}}s~\cite{AB19}.}
\label{fig-results-ilp}
\end{figure}

\subsection{Comparison with the ILP}

Almeida and Brás~\cite{AB19} provided four ILP variants for \textsc{Vertex Triangle~$2$-Club}.
In our experiments the variant \texttt{Fxkany} was the fastest of these four variants (see Figure~\ref{fig-results-ilp-variants} in the Appendix).
In Figure~\ref{fig-results-ilp} we compared our slowest variant \texttt{Basic} and our fastest variant \texttt{N-LB} with \texttt{Fxkany}.
For all~$\ell$, \texttt{Basic} is comparable with \texttt{Fxkany}, that is,
\texttt{Fxkany} and \texttt{Basic} solved roughly~$80\%$ of the instances in one hour. 
In comparison, our fastest variant \texttt{N-LB} solved~$80\%$ of the instances in less than one minute.
Thus, \texttt{N-LB} is more than 60~times faster than the previous best ILP for the vertex-variant.

\begin{table}[t!]
\caption{Average percentage of the total running time needed for preprocessing an instance for both the edge and the vertex-variant for different values of~$\ell$.
A value of NA indicates that \texttt{Fxkany} can only solve the vertex-variant. }
{\footnotesize
\begin{tabularx}{\textwidth}{p{1.5cm} X p{2cm} X X p{2cm} X}
  \toprule
\multirow{2}{*}{Algorithm} & \multicolumn{3}{c}{Vertex-Variant} & \multicolumn{3}{c}{Edge-Variant} \\
\cmidrule(l{0em}r{1em}){2-4}
\cmidrule(l{0em}r{1em}){5-7}

& $\ell\le 5$ & $6\le\ell\le 15$ & $\ell\ge 16$ & $\ell\le 5$ & $6\le\ell\le 15$ & $\ell\ge 16$ \\
\midrule
\texttt{N-LB}   & $73.0\%$ & $69.6\%$ & $77.6\%$ & $80.1\%$ & $87.4\%$ & $98.4\%$ \\
\texttt{Multi-LB} &$99.5\%$ & $98.4\%$ & $96.0\%$& $99.4\%$ & $99.0\%$ & $99.9\%$ \\
\texttt{Fxkany} & $15.0\%$ & $28.3\%$ & $59.0\%$ & NA& NA& NA \\
\bottomrule
\end{tabularx}
}\label{tab-preproc-times}
\end{table}

\begin{figure}[t!]
\centering
\begin{minipage}[b]{0.45\textwidth}
\includegraphics[width=\textwidth]{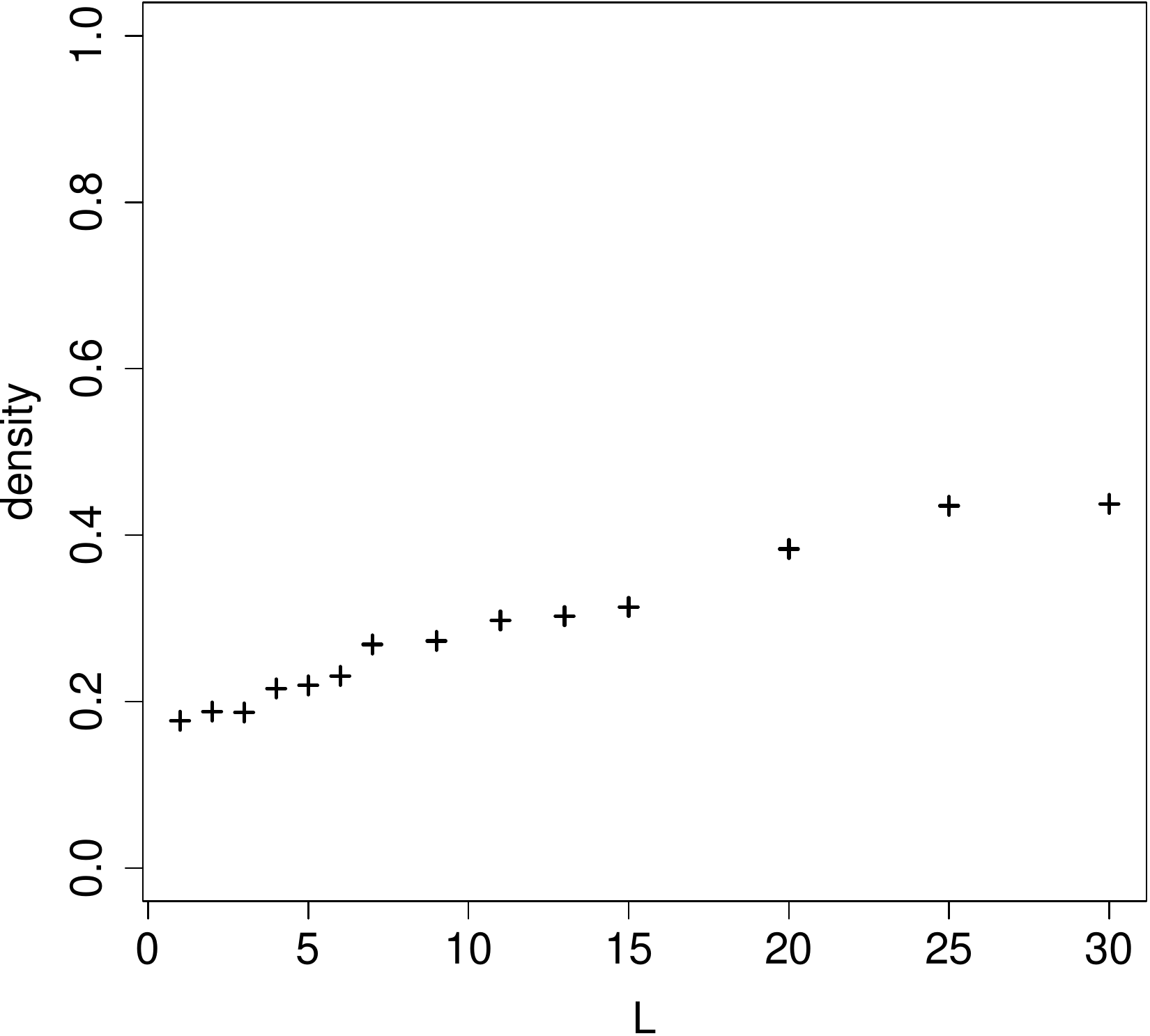}
\end{minipage}
\begin{minipage}[b]{0.45\textwidth}
\includegraphics[width=\textwidth]{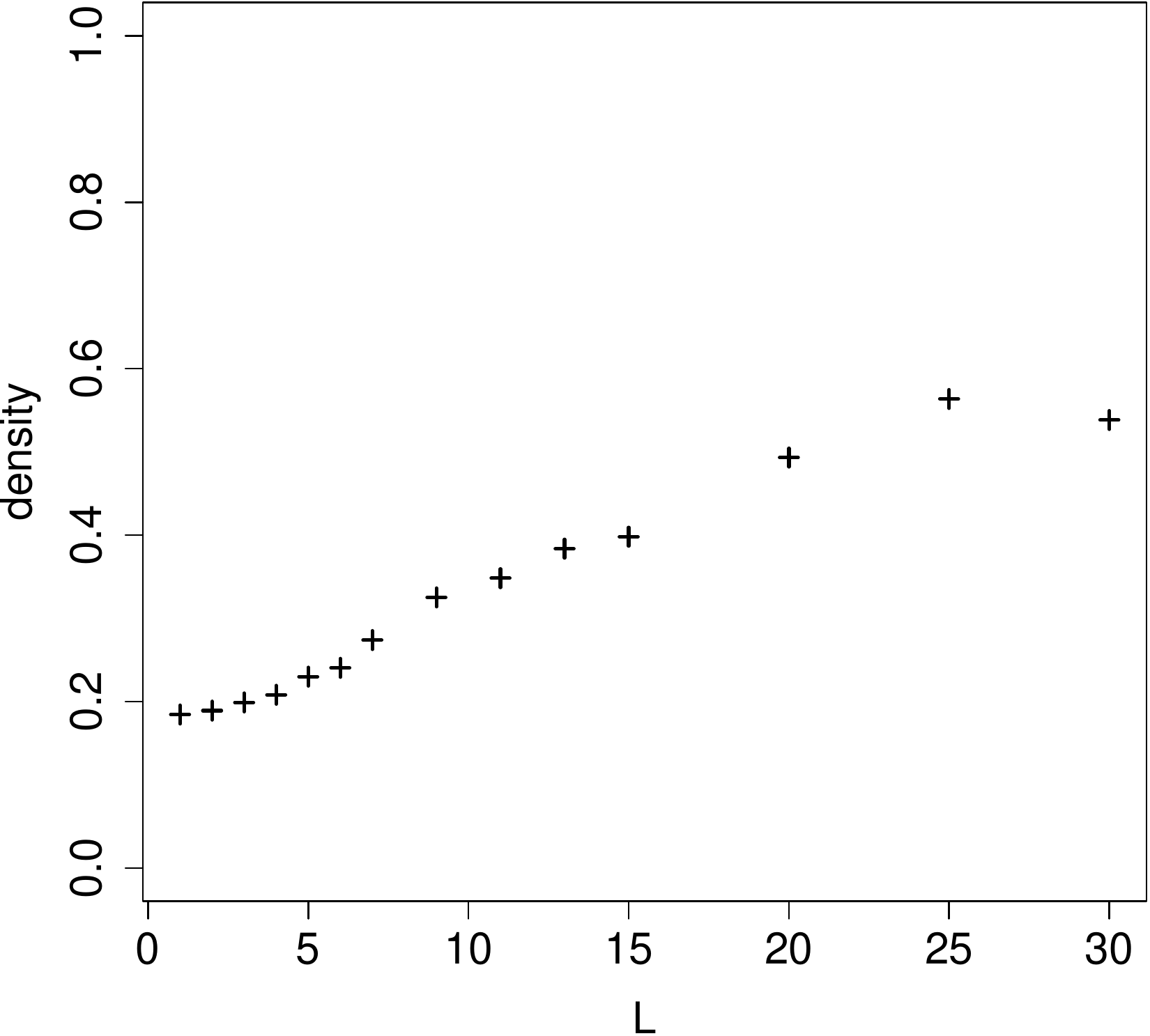}
\end{minipage}
\caption{Dependence of the density on~$\ell$ (L). 
The left plot shows the result for \textsc{Vertex Triangle~$2$-Club} and the right plot shows the result for \textsc{Edge Triangle~$2$-Club}.}
\label{fig-results-density}
\end{figure}

We also investigated how much time the algorithms spend for the preprocessing (see Table~\ref{tab-preproc-times}). 
For our algorithm the preprocessing (see also Figure~\ref{fig-algo-seqeunce}) comprises the initial application of the Basic Rules, the N-LB and the subsequent application of the 2-NR and the Basic Rules (\texttt{N-LB} and \texttt{Multi-LB}), and the G-LB and the subsequent application of the Basic Rules (\texttt{Multi-LB}).
For \texttt{Fxkany}, the preprocessing comprises the application of the LTR.
For our algorithm the average percentage of the total running time needed for preprocessing is at least~$2/3$ of the total running time.
In contrast, the average percentage of the total running time needed for preprocessing for \texttt{Fxkany} is very low for~$\ell\le 5$ but almost~$60\%$ for~$\ell\ge 15$.
One reason for this disparity is that our preprocessing comprises more rules and that our algorithm spends very little time for the branching.
Also, for our algorithm we did not observe big differences in the behaviour for the vertex and the edge-variant.

\subsection{Properties of the Solution}

\begin{figure}[t!]
\centering
\begin{minipage}[b]{0.45\textwidth}
\includegraphics[width=\textwidth]{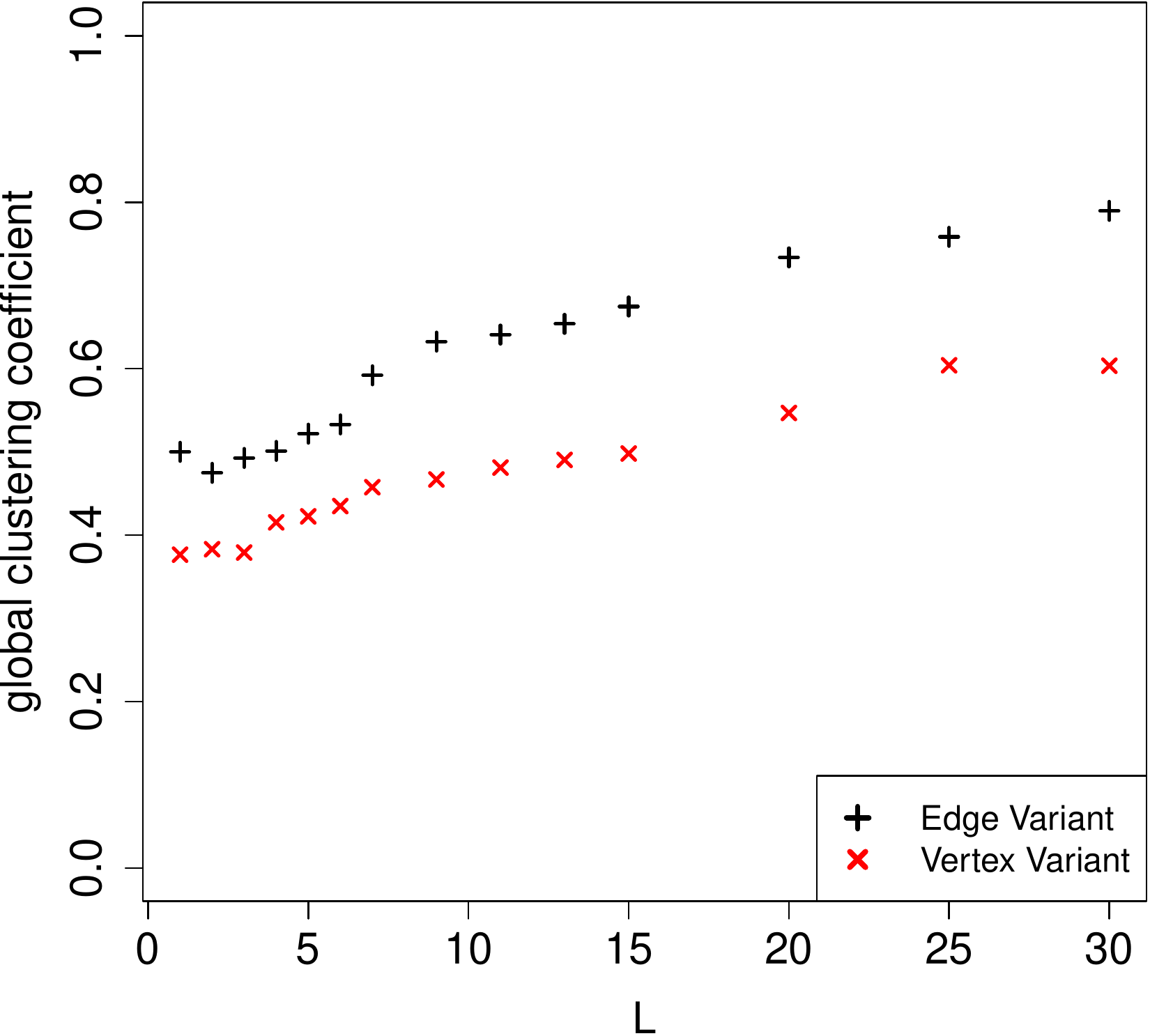}
\end{minipage}
\begin{minipage}[b]{0.45\textwidth}
\includegraphics[width=\textwidth]{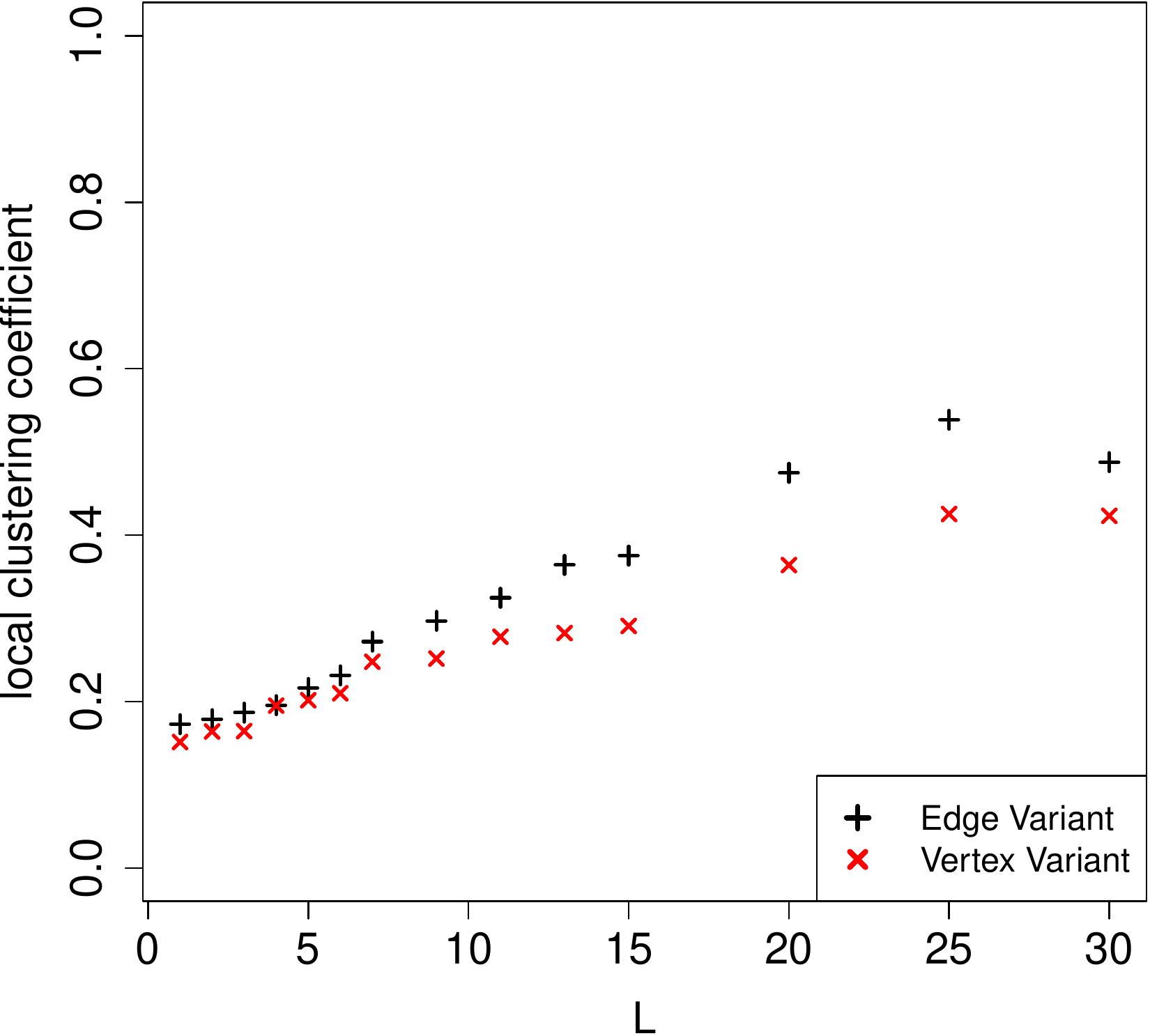}
\end{minipage}
\caption{Dependence of the average global and the smallest local clustering coefficients on~$\ell$ (L). Black +-signs show the values for the edge-variant, red crosses show the values for the vertex-variant.}
\label{fig-results-cluster-coeff-edge}
\end{figure}

We also considered the impact of~$\ell$ on the density. 
Figure~\ref{fig-results-density} shows the average densities of the solutions for all instances that were solved within the time limit and had non-empty solutions for all~$\ell\le 30$. 
In general, the density grows with~$\ell$. 
Already for~$\ell=1$, the density is relatively high.
One can see that the density grows more strongly with increasing~$\ell$ in the edge-variant than in the vertex-variant.
The main reason for this is that the minimum degree in an edge-$\ell$-triangle~$2$-club is higher than the minimum degree in a vertex-$\ell$-triangle~$2$-club, especially for~$\ell\ge 10$.

Finally, we consider the impact of~$\ell$ on global and smallest local clustering coefficients for the solution. 
Figure~\ref{fig-results-cluster-coeff-edge} shows the average coefficients for all instances that were solved within the time limit and had non-empty solutions for all~$\ell\le 30$. 
In general, both coefficients grow with~$\ell$. 
As expected, the growth is more rapid for the stricter edge-variant where the minimum degree is~$\ell+1$. 
Observe that already for~$\ell=1$, the global clustering coefficient is relatively high for both variants. 
Naturally, the smallest local clustering is comparably smaller but achieves high values already for~$\ell\le 10$ in both variants. 
Summarizing, this shows that the computed solutions fulfill further desirable community properties.

\section{Conclusion}

We provided an exact branch-and-bound solver for \textsc{Vertex Triangle~$2$-Club} and for \textsc{Edge Triangle~$2$-Club}.
We showed that our solver outperforms existing ILPs by Almeida and Brás~\cite{AB19} for \textsc{Vertex Triangle~$2$-Club}.
Furthermore, we showed that the local and global clustering coefficient in a triangle~$2$-club in real-world instances is much higher than the guaranteed local or global clustering coefficient by the definition of this model (see~\cite{AB19}).
Also, our experiments showed that both the local and global clustering coefficient is higher in the edge-variant than the corresponding values in the vertex-variant.
Since also edge-$\ell$-triangle~$2$-clubs can be found faster than vertex-$\ell$-triangle~$2$-clubs (which is mainly because of the higher minimum degree in the edge-variant), we conclude that the edge-$\ell$-triangle~$2$-club model is not only preferable from a modeling point of view~\cite{GKS22}, but also regarding the running time and quality of the solution.

It would be interesting to also develop ILP formulations for \textsc{Edge Triangle~$2$-Club} and compare them with our algorithm.
For this problem, we expect ILP-based formulations to be slower than for \textsc{Vertex Triangle~$2$-Club}:  direct formulations need variables that encode the presence of an edge which is not necessary for the vertex-triangle property.
It seems very promising to tune such an ILP with the LDR (Rule~\ref{rr-remove-vertices-to-low-degree}) and the LTR (Rule~\ref{rr-remove-vertices-to-low}) as an initial pre-processing as it was done by Almeida and Brás for \textsc{Vertex Triangle~$2$-Club} (they only used the LTR).

Finally, it would be interesting to lift our implementation to be able to also solve \textsc{Vertex Triangle~$s$-Club} and \textsc{Edge Triangle~$s$-Club} for~$s\ge 3$.
Some reduction rules like the LTR or the LDR can be used directly for these problems.
To use the LCR, one simply needs to adapt the definition of incompatibility.
But not every reduction rule can be adapted that easily:
For example the NCR (Rule~\ref{rr-no-choice2}) is not true anymore since two non-adjacent marked vertices~$u$ and~$w$ do not need to have a common neighbor anymore since~$s\ge 3$.
Another challenge is to develop good heuristics: 
For~$s\ge 3$, a lower bound based only on the neighborhood of a vertex might have a too large difference to the value of an optimal solution.
Finally,  one bottleneck in our solver for~$s=2$ is the check whether a vertex set is a $2$-club.
For~$s\ge 3$ this check needs even more time since it is not sufficient anymore to check whether each pair of vertices is adjacent or has a common neighbor.
Hence, it is important to check the $s$-club property sufficiently fast.

\bibliographystyle{plainurl}

\newpage
\appendix

\section{Comparison of the ILP Variants for Vertex Triangle~2-Club}

\begin{figure}[h!]
\centering
\begin{minipage}[b]{0.45\textwidth}
\includegraphics[width=\textwidth]{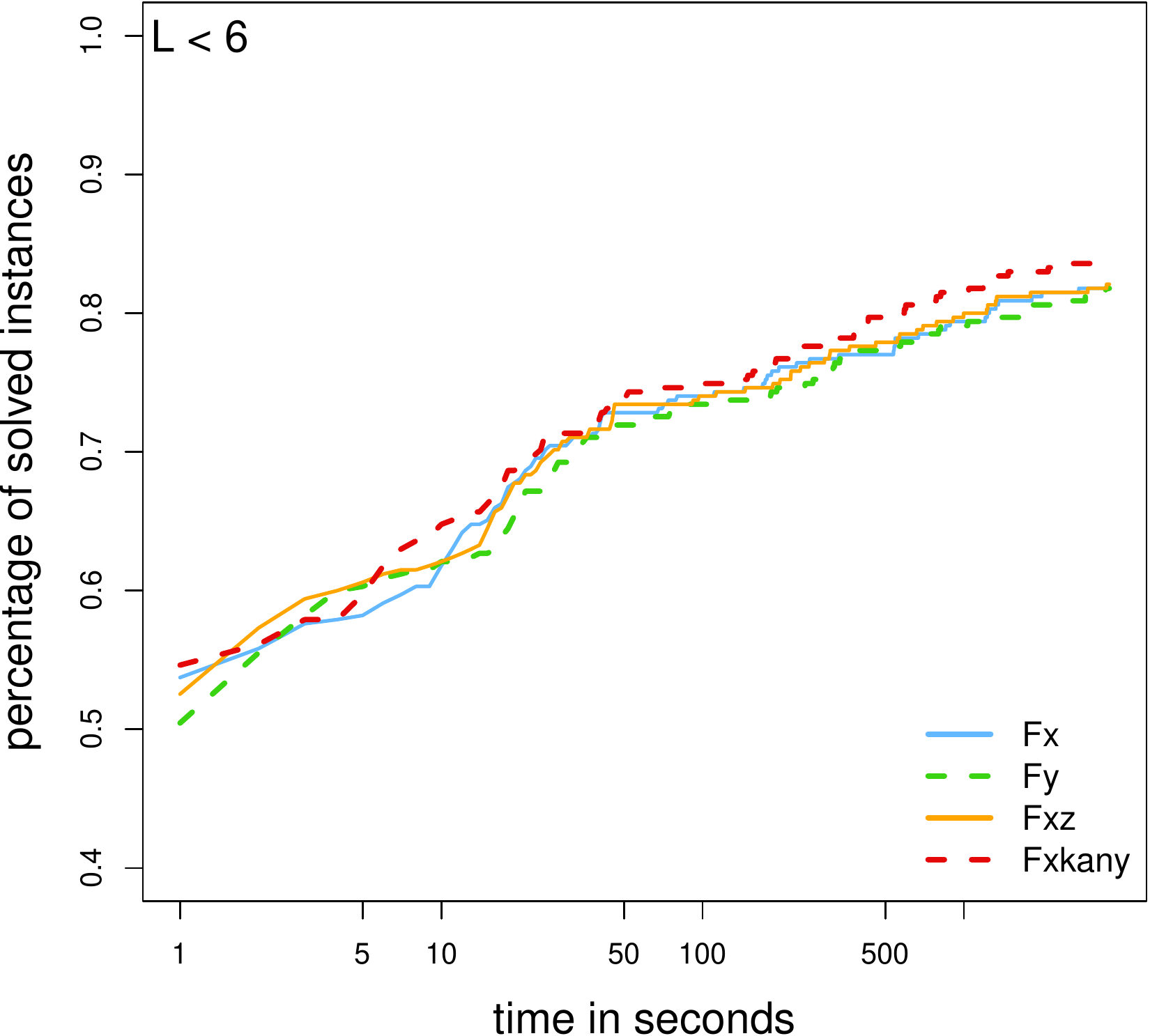}
\end{minipage}
\begin{minipage}[b]{0.45\textwidth}
\includegraphics[width=\textwidth]{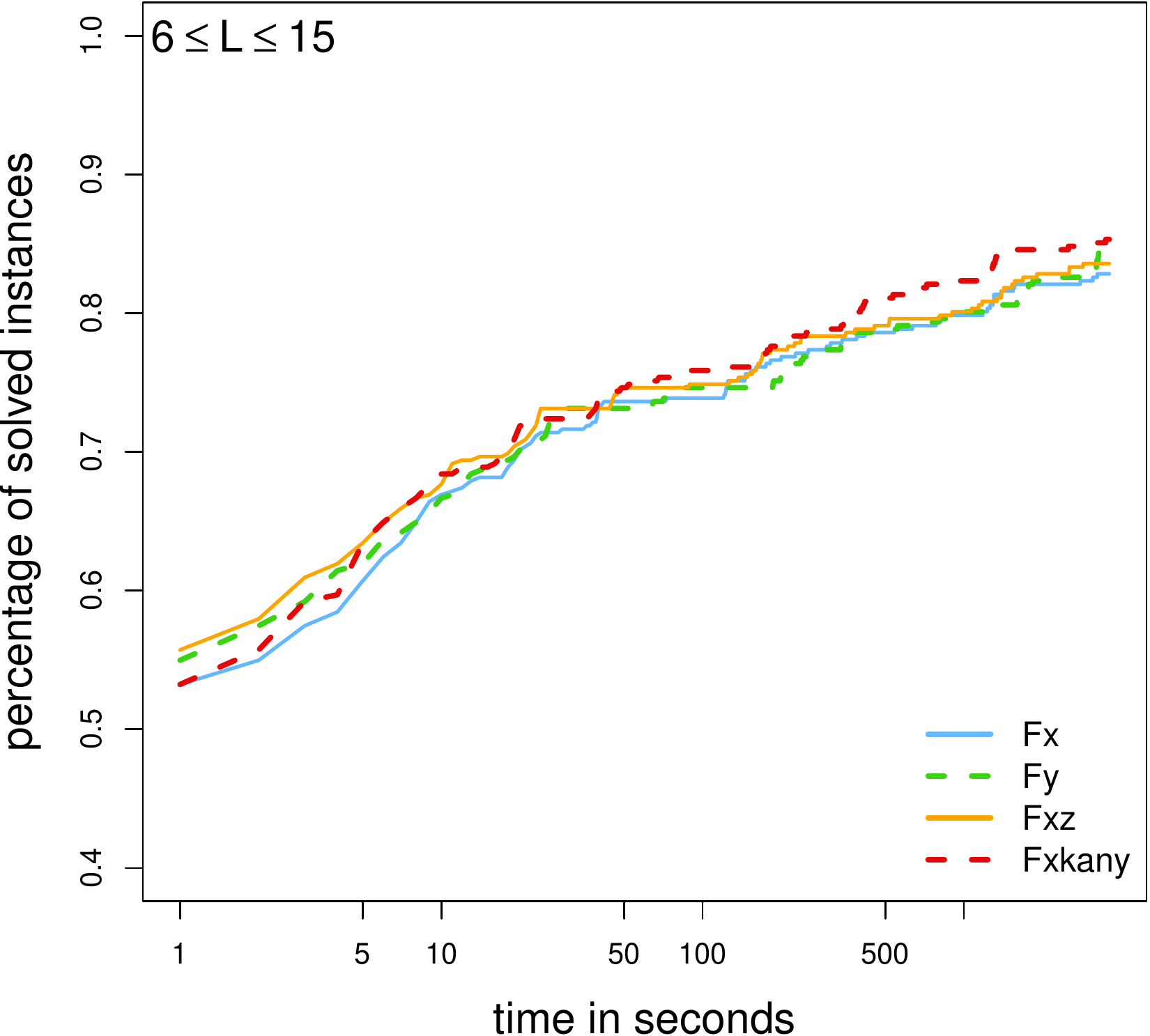}
\end{minipage}
\begin{minipage}[b]{0.45\textwidth}
\includegraphics[width=\textwidth]{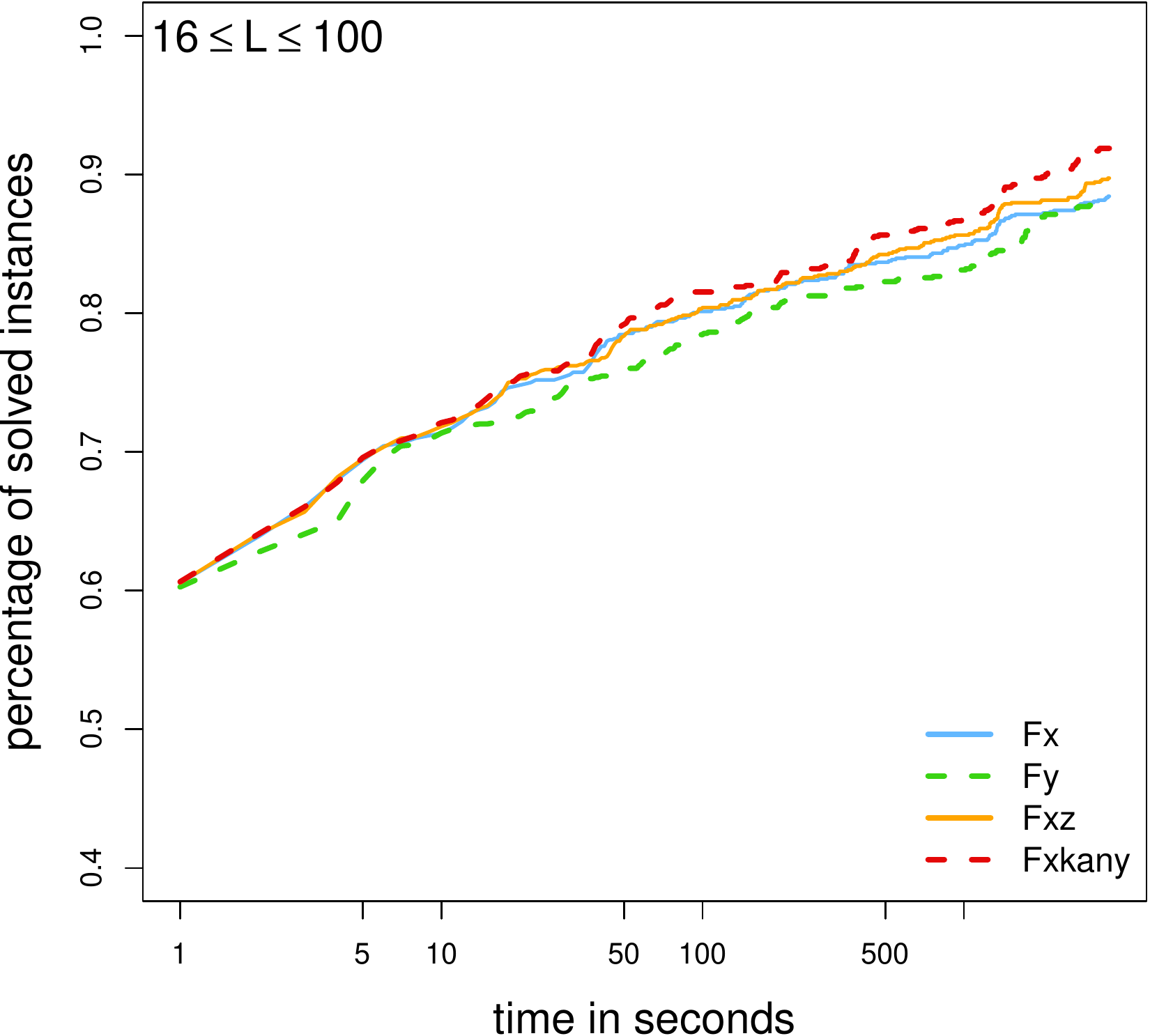}
\end{minipage}
\caption{Comparison of the four ILP variants of the ILP of Almeida and Br{\'{a}}s~\cite{AB19} for \textsc{Vertex Triangle~$2$-Club}.
}
\label{fig-results-ilp-variants}
\end{figure}

\end{document}